\documentclass{article} 
\usepackage{iclr2023_conference,times}

\usepackage{hyperref}       
\usepackage{url}            
\usepackage{booktabs}       
\usepackage{amsfonts}       
\usepackage{nicefrac}       
\usepackage{microtype}      
\usepackage{xcolor}         
\usepackage{graphicx}

\usepackage{wrapfig}
\usepackage{slashbox}
\newcommand{\indep}{\perp \!\!\! \perp}
\newcommand{\indepnot}{\not\!\perp\!\!\!\perp}
\newcommand{\E}{\mathbb{E}}
\newcommand{\R}{\mathbb{R}}
\newcommand{\N}{\mathbb{N}}

\usepackage{xspace}
  
  \newcommand{\ie}{i.\,e.\xspace}
\usepackage{enumitem}
\usepackage{bbm}
\usepackage{setspace}
\newcommand{\frameworkname}{MRIV\xspace}
\newcommand{\modelname}{\mbox{MRIV-Net}\xspace}

\usepackage{amssymb,amsmath,amsthm}
\usepackage{algorithm2e}
\RestyleAlgo{ruled}
\theoremstyle{definition}
\newtheorem{assumption}{Assumption}
\newtheorem{definition}{Definition}
\theoremstyle{plain}
\newtheorem{lemma}{Lemma}
\newtheorem{theorem}{Theorem}

\title{Estimating individual treatment effects under unobserved confounding using binary instruments}

\author{Dennis Frauen\textsuperscript{\rm 1, 2} \& Stefan Feuerriegel\textsuperscript{\rm 1, 2} \\
\textsuperscript{\rm 1} Munich Center for Machine Learning\\
\textsuperscript{\rm 2} LMU Munich
}

\iclrfinalcopy 
\begin{document}

\maketitle

\begin{abstract}
Estimating conditional average treatment effects (CATEs) from observational data is relevant in many fields such as personalized medicine. However, in practice, the treatment assignment is usually confounded by unobserved variables and thus introduces bias. A remedy to remove the bias is the use of instrumental variables (IVs). Such settings are widespread in medicine (e.g., trials where the treatment assignment is used as binary IV). In this paper, we propose a novel, multiply robust machine learning framework, called \frameworkname, for estimating CATEs using binary IVs and thus yield an unbiased CATE estimator. Different from previous work for binary IVs, our framework estimates the CATE directly via a pseudo outcome regression. (1)~We provide a theoretical analysis where we show that our framework yields multiple robust convergence rates: our CATE estimator achieves fast convergence even if several nuisance estimators converge slowly. (2)~We further show that our framework asymptotically outperforms state-of-the-art plug-in IV methods for CATE estimation, in the sense that it achieves a faster rate of convergence if the CATE is smoother than the individual outcome surfaces. (3)~We build upon our theoretical results and propose a tailored deep neural network architecture called \modelname for CATE estimation using binary IVs. Across various computational experiments, we demonstrate empirically that our \modelname achieves state-of-the-art performance. To the best of our knowledge, our \frameworkname is the first multiply robust machine learning framework tailored to estimating CATEs in the binary IV setting. 
\end{abstract}

\section{Introduction}


Conditional average treatment effects (CATEs) are relevant across many disciplines such as marketing \citep{Varian.2016} and personalized medicine \citep{Yazdani.2015}. Knowledge about CATEs provides insights into the heterogeneity of treatment effects, and thus helps in making potentially better treatment decisions \citep{Frauen.2022b}. 


Many recent works that use machine learning to estimate causal effects, in particular CATEs, are based on the assumption of unconfoundedness \citep{Alaa.2017, Lim.2018, Melnychuk.2022, Melnychuk.2022b}. In practice, however, this assumption is often violated because it is common that some confounders are not reported in the data. Typical examples are income or the socioeconomic status of patients, which are not stored in medical files. If the confounding is sufficiently strong, standard methods for estimating CATEs suffer from confounding bias \citep{Pearl.2009}, which may lead to inferior treatment decisions. 

To handle unobserved confounders, instrumental variables (IVs) can be leveraged to relax the assumption of unconfoundedness and still compute reliable CATE estimates. IV methods were originally developed in economics \citep{Wright.1928}, but, only recently, there is a growing interest in combining IV methods with machine learning (see Sec.~\ref{sec:related_work}). Importantly, IV methods outperform classical CATE estimators if a sufficient amount of confounding is not observed \citep{Hartford.2017}. We thus aim at estimating CATEs from observational data under unobserved confounding using IVs.


In this paper, we consider the setting where a single binary instrument is available. This setting is widespread in personalized medicine (and other applications such as marketing or public policy) \citep{Bloom.1997}. In fact, the setting is encountered in essentially all observational or randomized studies with observed non-compliance \citep{Imbens.1994}. As an example, consider a randomized controlled trial (RCT), where treatments are randomly assigned to patients and their outcomes are observed. Due to some potentially unobserved confounders (e.g., income, education), some patients refuse to take the treatment initially assigned to them. Here, the treatment assignment serves as a binary IV. Moreover, such RCTs have been widely used by public decision-makers, e.g., to analyze the effect of health insurance on health outcome (see the so-called \emph{Oregon health insurance experiment}) \citep{Finkelstein.2012} or the effect of military service on lifetime earnings \citep{Angrist.1990}.


We propose a novel machine learning framework (called \frameworkname) for estimating CATEs using binary IVs. Our framework takes an initial CATE estimator and nuisance parameter estimators as input to perform a pseudo-outcome regression. Different to existing literature, our framework is \textbf{multiply robust}\footnote{For a detailed introduction to multiple robustness and its importance in treatment effect estimation, we refer to \citep{Wang.2018}, Section 4.5.}, i.e., we show that it is consistent in the union of three different model specifications. This is different from existing methods for CATE estimation using IVs such as \citet{Okui.2012}, \citet{Syrgkanis.2019}, or plug-in estimators \citep{BargagliStoffi.2021, Imbens.1994}.

We provide a theoretical analysis, where we use tools from \citet{Kennedy.2022d} to show that our framework achieves a multiply robust convergence rate, i.e., our \frameworkname converges with a fast rate even if several nuisance parameters converge slowly. We further show that, compared to existing plug-in IV methods, the performance of our framework is asymptotically superior. Finally, we leverage our framework and, on top of it, build a tailored deep neural network called \modelname. 

\textbf{Contributions:} (1)~We propose a novel, multiply robust machine learning framework (called \frameworkname) to learn the CATE using the binary IV setting. To the best of our knowledge, ours is the first that is shown to be multiply robust, i.e., consistent in the union of three model specifications. For comparison, existing works for CATE estimation only show double robustness \citep{Wang.2018, Syrgkanis.2019}. (2)~We prove that \frameworkname achieves a multiply robust convergence rate. This is different to methods for IV settings which do not provide robust convergence rates \citep{Syrgkanis.2019}. We further show that our \frameworkname is asymptotically superior to existing plug-in estimators. (3)~We propose a tailored deep neural network, called \modelname, which builds upon our framework to estimate CATEs . We demonstrate that \modelname achieves state-of-the-art performance.

\section{Problem setup}

\textbf{Data generating process:}
We observe data $\mathcal{D} = (x_i, z_i, a_i, y_i)_{i=1}^n$ consisting of $n \in \N$ observations of the tuple $(X, Z, A, Y)$.
Here, $X \in \mathcal{X}$ are observed confounders, $Z \in \{0, 1\}$ is a binary instrument, $A \in \{0, 1\}$ is a binary treatment, and $Y \in \R$ is an outcome of interest. Furthermore, we assume the existence of unobserved confounders $U \in \mathcal{U}$, which affect both the treatment $A$ and the outcome $Y$.
\begin{wrapfigure}{r}{0.25\textwidth}
\vspace{-0.55cm}
\begin{center}
\includegraphics[width=0.25\textwidth]{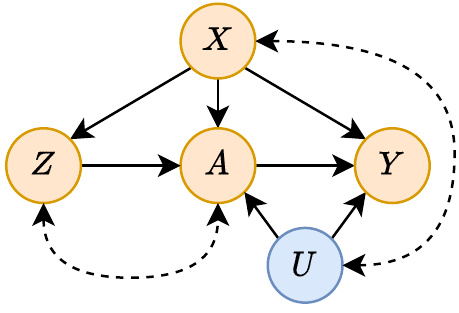}
\end{center}\vspace{-0.4cm}
\caption{Underlying causal graph. The instrument $Z$ has a direct influence on the treatment $A$, but does not have a direct effect on the outcome $Y$. Note that we allow for unobserved confounders for both $Z$--$A$ (dashed line) and $A$--$Y$ (given by $U$).}
\label{fig:causal_graph}
\vspace{-1.2cm}
\end{wrapfigure}
The causal graph is shown in Fig.~\ref{fig:causal_graph}.

\textbf{Applicability}: Our proposed framework is widely applicable in practice, namely to all settings with the above data generating process. This includes both (1)~observational data and (2)~RCTs with non-compliance. For (1), observational data is commonly encountered in, e.g., personalized medicine. Here, modeling treatments as binary variables is consistent with previous literature on causal effect estimation and standard in medical practice \citep{Robins.2000}. For (2), our setting is further encountered in RCTs when the instrument $Z$ is a randomized treatment assignment but individuals do not comply with their treatment assignment. Such RCTs have been extensively used by public decision-makers, e.g., to analyze the effect of health insurance on health outcome \citep{Finkelstein.2012} or the effect of military service on lifetime earnings \citep{Angrist.1990}. 

We build upon the potential outcomes framework \citep{Rubin.1974} for modeling causal effects. Let $Y(a,z)$ denote the potential outcome that would have been observed under $A=a$ and $Z=z$. Following previous literature on IV estimation \citep{Wang.2018}, we impose the following standard IV assumptions on the data generating process.

\begin{assumption}[Standard IV assumptions \citep{Wang.2018, Wooldridge.2013}]
\label{ass:iv}
We assume: (1)~\emph{Exclusion:} $Y(a,z) = Y(a)$ for all $a, z \in \{0,1\}$, i.e., the instrument has no direct effect on the patient outcome; (2)~\emph{Independence:} $Z \indep U \mid X$; (3)~\emph{Relevance:} $Z \indepnot A \mid X$, (iv)~\emph{The model includes all $A$--$Y$ confounder:} $Y(a) \indep (A, Z) \mid (X, U)$ for all $a \in \{0,1\}$.
\end{assumption}

Assumption 1 is standard for IV methods and fulfilled in practical settings where IV methods are applied \citep{Angrist.1990, Angrist.1991, Imbens.1994}. Note that Assumption~\ref{ass:iv} does not prohibit the existence of unobserved $Z$--$A$ confounders. On the contrary, it merely prohibits the existence of unobserved counfounders that affect all $Z$, $A$, and $Y$ simultaneously, as it is standard in IV settings \citep{Wooldridge.2013}. A practical and widespread example where Assumption~\ref{ass:iv} is satisfied are randomized controlled trials~(RCTs) with non-compliance \citep{Imbens.1994}. Here, the treatment assignment $Z$ is randomized, but the actual relationship between treatment $A$ and outcome $Y$ may still be confounded. For instance, in the \emph{Oregon health insurance experiment} \citep{Finkelstein.2012}, people were given access to health insurance ($Z$) by a lottery with aim to study the effect of health insurance ($A$) on health outcome ($Y$) \citep{Finkelstein.2012}. Here, the lottery winners needed to sign up for health insurance and thus both $Z$ and $A$ are observed.
\textbf{Objective:}
In this paper, we are interested in estimating the \emph{conditional average treatment effect} (CATE)
\begin{equation}
 \tau(x) = \E[Y(1) - Y(0) \mid X = x].
\end{equation}
If there is no unobserved confounding ($U = \emptyset$), the CATE is identifiable from observational data \citep{Shalit.2017}. However, in practice, it is often unlikely that all confounders are observable. To account for this, we leverage the instrument $Z$ to identify the CATE. We state the following assumption for identifiability.

\begin{assumption}[Identifiability of the CATE \citep{Wang.2018}]\label{ass:identification}
 At least one of the following two statements holds true: (1)~$\E[A \mid Z=1, X, U] - \E[A \mid Z=0, X, U] = \E[A \mid Z=1, X] - \E[A \mid Z=0, X]$;
 or (2)~$\E[Y(1) - Y(0) \mid X, U] = \E[Y(1) - Y(0) \mid X]$. 
\end{assumption}
\textbf{Example:} Assumption \ref{ass:iv} holds when the function $f(a, X, U) = \E[Y(a) \mid X, U]$ is additive with respect to $a$ and $U$, e.g., $f(a,X, U) = g(a, X) + h(U)$ for measurable functions $h$ and $g$. This implies that no unobserved confounder affects the outcome through a path which is also affected by the treatment. For example, with patient income as unobserved confounder, the treatment should not affect the (future) patient income.

Under Assumptions \ref{ass:iv} and \ref{ass:identification}, the CATE is identifiable \citep{Wang.2018}. It can be written as
\begin{equation}
\label{eq:identification}
  \tau(x) = \frac{\mu_1^Y(x) - \mu_0^Y(x)}{\mu_1^A(x) - \mu_0^A(x)} = \frac{\delta_Y(x)}{\delta_A(x)},
\end{equation}
where $\mu_i^Y(x) = \E[Y \mid Z = i, X = x]$ and $\mu_i^A(x) = \E[A \mid Z = i, X = x]$. Even if Assumption~\ref{ass:identification} does not hold, all our results in this paper still hold for the quantity on the right-hand side of Eq.~\eqref{eq:identification}. In certain cases, this quantity still allows for interpretation: If no unobserved $Z$--$A$ confounders exist, it can be interpreted as conditional version of the \emph{local average treatment effect}~(LATE) \citep{Imbens.1994} under a monotonicity assumption. Furthermore, under a no-current-treatment-value-interaction assumption, it can be interpreted as conditional \emph{treatment effect on the treated}~(ETT) \citep{Wang.2018}.\footnote{The conditional LATE measures the CATE for individuals which are part of the complier subpopulation, i.e., for whom $A(Z=1) > A(Z=0)$. The conditional ETT measures the CATE for treated individuals.} This has an important implication for our results: If Assumption~\ref{ass:identification} does not hold in practice, our estimates still provide conditional LATE or ETT estimates under the respective assumptions because they are based on Eq.~\eqref{eq:identification}. If Assumption~\ref{ass:identification} does hold, all three -- i.e., CATE, conditional LATE, and ETT -- coincide \citep{Wang.2018}.

\section{Related work}
\label{sec:related_work}

\textbf{Machine learning methods for IV:}
Only recently, machine learning has been integrated into IV methods. These are: \citet{Singh.2019} and \citet{Xu.2021} generalize 2SLS by learning complex feature maps using kernel methods and deep learning, respectively. \citet{Hartford.2017} adopts a two-stage neural network architecture that performs the first stage via conditional density estimation. \citet{Bennett.2019} leverages moment conditions for IV estimation. However, the aforementioned methods are not specifically designed for the binary IV setting but, rather, for multiple IVs or treatment scenarios. In particular, they impose stronger assumptions such as additive confounding in order to identify the CATE. Note that additive confounding is a special case of our Assumption~\ref{ass:identification}. Moreover, they do not have robustness properties. In the binary IV setting, current methods proceed by estimating $\mu_i^Y(x)$ and $\mu_i^A(x)$ separately, before plugging them in Eq.~\ref{eq:identification} \citep{Imbens.1994, Angrist.1996, BargagliStoffi.2021}. As a result, these suffer from plug-in bias and do \emph{not} offer robustness properties.

\begin{wraptable}{r}{7cm}
\vspace{-1.2cm}
\caption{Key methods for causal effect estimation with IVs and associated robustness properties}
\label{tab:rel_work}
\resizebox{7cm}{!}{
\begin{tabular}{lcc}
\toprule
\backslashbox{Robustness}{Estimand} & {ATE} & {ITE}  \\
\midrule
Doubly robust & \citet{Okui.2012} & \citet{Syrgkanis.2019}\\
Multiply robust &\citet{Wang.2018}&  \textbf{{MRIV}} (ours) \\
\bottomrule
\end{tabular}
}
\vspace{-0.4cm}
\end{wraptable}

\textbf{Doubly robust IV methods:} Recently, doubly robust methods have been proposed for IV settings: \citet{Kennedy.2019} propose a pseudo regression estimator for the local average treatment effect using continuous instruments, which has been extended to conditional effects by \citet{Semenova.2021}. Furthermore, \citet{Singh.2019b} use a doubly robust approach to estimate average compiler parameters. Finally, \citet{Ogburn.2015} and \citet{Syrgkanis.2019} propose doubly robust CATE estimators in the IV setting which both rely on doubly robust parametrizations of the uncentered efficient influence function \citep{Okui.2012}. However, none of these estimators has been shown to be multiply robust in the sense that they are consistent in the union of more than two model specifications \citep{Wang.2018}.

\textbf{Multiply robust IV methods:} Multiply robust estimators for IV settings have been proposed only for average treatment effects (ATEs) \citep{Wang.2018} and optimal treatment regimes \citep{Cui.2021} but \underline{not} for CATEs . In particular, \citet{Wang.2018} derive a multiply robust parametrization of the efficient influence function for the ATE. However, there exists \underline{no} method that leverages this result for CATE estimation. 
We provide a detailed, technical comparison of existing methods and our framework in Appendix~\ref{app:baseline}.

\textbf{Doubly robust rates for CATE estimation:} \citet{Kennedy.2022d} analyzed the doubly robust learner in the standard (non-IV) setting and derived doubly robust convergence rates. However, Kennedy's result is not applicable in the IV setting, because we use the multiply robust parametrization of the efficient influence function from \citet{Wang.2018}. In our paper, we rely on certain results from Kennedy, but use these to derive of a multiply robust rate. In particular, this required the derivation of the bias term for a larger number of nuisance parameters (see Appendix~\ref{app:proofs}).

\textbf{Research gap:} To the best of our knowledge, there exists no method for CATE estimation under unobserved confounding that has been shown to be \emph{multiply robust}. To fill this gap, we propose \frameworkname: a \emph{multiply robust} machine learning framework tailored to the binary IV setting. For this, we build upon the approach by \citet{Kennedy.2022d} to derive robust convergence rates, yet this approach has not been adapted to IV settings, which is our contribution.

\section{\frameworkname for estimating CATEs using binary instruments}

In the following, we present our \frameworkname framework for estimating CATEs under unobserved confounding (Sec.~\ref{sec:methodology}). We then derive an asymptotic convergence rate for \frameworkname (Sec.~\ref{sec:theory}) and finally use our framework to develop a tailored deep neural network called \modelname (Sec.~\ref{sec:neural_nets}). 

\vspace{-0.2cm}

\subsection{Framework}
\label{sec:methodology}

\textbf{Motivation:} A na{\"i}ve approach to estimate the CATE is to leverage the identification result in Eq.~\eqref{eq:identification}. Assuming that we have estimated the nuisance components $\hat{\mu}_i^Y$ and $\hat{\mu}_i^A$ for $i \in \{0,1\}$, we can simply plug them into Eq.~\eqref{eq:identification} to obtain the so-called (plug-in) Wald estimator $\hat{\tau}_{\mathrm{W}}(x)$ \citep{Wald.1940}.

However, in practice, the true CATE curve $\tau(x)$ is often simpler (e.g., smoother, more sparse) than its complements $\mu_i^Y(x)$ or $\mu_i^A(x)$ \citep{Kunzel.2019}. In this case, $\hat{\tau}_{\mathrm{W}}(x)$ is inefficient because it models all components separately, and, to address this, our proposed framework estimates $\tau$ {directly} using a pseudo outcome regression. 

\textbf{Overview:} We now propose \frameworkname. \frameworkname is a two-stage meta learner that takes any base method for CATE estimation as input. For instance, the base method could be the Wald estimator from Eq.~\eqref{eq:identification}, any other IV method such as 2SLS, or a deep neural network (as we propose in our \modelname later in Sec.~\ref{sec:neural_nets}). In Stage 1, \frameworkname produces nuisance estimators $\hat{\mu}_0^Y(x)$, $\hat{\mu}_0^A(x)$, $\hat{\delta}_A(x)$, and $\hat{\pi}(x)$, where $\hat{\pi}(x)$ is an estimator of the propensity score $\pi(x) = \mathbb{P}(Z = 1 \mid X = x)$. In Stage 2, \frameworkname estimates $\tau(x)$ directly using a pseudo outcome $\hat{Y}_{\mathrm{\mathrm{MR}}}$ as a regression target. 

Given an arbitrary initial CATE estimator $\hat{\tau}_{\mathrm{init}}(x)$ and nuisance estimates $\hat{\mu}_0^Y(x)$, $\hat{\mu}_0^A(x)$, $\hat{\delta}_A(x)$, and $\hat{\pi}(x)$, we define the pseudo outcome

\vspace{-0.4cm}
\begin{equation}
\label{eq:pseudo_outcome}
    \resizebox{.6\hsize}{!}{$\hat{Y}_{\mathrm{\mathrm{MR}}} = \left(\frac{Z -(1-Z)}{\hat{\delta}_A(X)}\right) \left( \frac{Y - \left( \hat{\mu}_0^Y(X) + \hat{\tau}_{\mathrm{init}}(X) \, (A - \hat{\mu}_0^A(X)) \right) }{Z \, \hat{\pi}(X) + (1-Z) (1-\hat{\pi}(X))}\right) + \hat{\tau}_{\mathrm{init}}(X)$}.
\end{equation}
\vspace{-0.4cm}

\begin{wrapfigure}{L}{0.5\textwidth}
\begin{minipage}{0.5\textwidth}
\vspace{-0.5cm}
\begin{algorithm}[H]
\DontPrintSemicolon
\caption{\frameworkname}
\label{alg:mr}
\scriptsize
\SetKwInOut{Input}{Input}
\Input{~data $(X, Z, A, Y)$, initial CATE estimator $\hat{\tau}_{\mathrm{init}}(x)$}
\tcp{Stage 1: Estimate nuisance components}
$\hat{\pi}(x) \gets \hat{\E}[Z \mid X = x] $, \quad $\hat{\mu}_0^Y(x) \gets \hat{\E}[Y \mid X = x, Z = 0]$, \quad  $\hat{\mu}_0^A(x) \gets \hat{\E}[A \mid X = x, Z = 0]$\;
$\hat{\delta}_A(x) \gets \hat{\E}[A \mid X = x, Z = 1] - \hat{\E}[A \mid X = x, Z = 0]$\;
\tcp{Stage 2: Pseudo outcome regression}
$\hat{Y}_{\mathrm{\mathrm{MR}}} \gets \left( \frac{Z -(1-Z)}{\hat{\delta}_A(X)} \right) \left(\frac{Y - A \, \hat{\tau}_{\mathrm{init}}(X) - \hat{\mu}_0^Y(X) + \hat{\mu}_0^A(X) \, \hat{\tau}_{\mathrm{init}}(X)}{Z \, \hat{\pi}(X) + (1-Z) (1-\hat{\pi}(X))} \right) +\hat{\tau}_{\mathrm{init}}(X)$\;
$\hat{\tau}_{\mathrm{MRIV}}(x) \gets \hat{\E}[\hat{Y}_{\mathrm{\mathrm{MR}}} \mid X = x] $
\end{algorithm}
\vspace{-0.5cm}
\end{minipage}
\end{wrapfigure}

The pseudo outcome $\hat{Y}_{\mathrm{\mathrm{MR}}}$ in Eq.~\eqref{eq:pseudo_outcome} is a multiply robust parameterization of the (uncentered) efficient influence function for the average treatment effect $\E_X[\tau(X)]$ (see the derivation in \citep{Wang.2018}). 
Once we have obtained the pseudo outcome $\hat{Y}_{\mathrm{\mathrm{MR}}}$, we regress it on $X$ to obtain the Stage 2 \frameworkname estimator $\hat{\tau}_{\mathrm{\frameworkname}}(x)$ for $\tau(x)$. The pseudocode for \frameworkname is given in Algorithm~\ref{alg:mr}. \frameworkname can be interpreted as a way to remove plug-in bias from $\hat{\tau}_{\mathrm{init}}(x)$ \citep{Curth.2020}. Using the fact that $\hat{Y}_{\mathrm{\mathrm{MR}}}$ is a multiply robust parametrization of the efficient influence function, we derive a multiple robustness property of $\hat{\tau}_{\mathrm{MRIV}}(x)$.

\begin{theorem}[multiple robustness property]\label{thrm:robustness}
Let $\hat{\mu}_0^Y(x)$, $\hat{\mu}_0^A(x)$, $\hat{\delta}_A(x)$, $\hat{\pi}(x)$, and $\hat{\tau}_{\mathrm{init}}(x)$ denote estimators of $\mu_0^Y(x)$, $\mu_0^A(x)$, $\delta_A(x)$, $\pi(x)$, and $\tau(x)$, respectively. Then, for all $x \in \mathcal{X}$, it holds that $\E[\hat{Y}_{\mathrm{\mathrm{MR}}} \mid X = x] = \tau(x)$,if least one of the following conditions is satisfied: (1)~$\hat{\mu}_0^Y = \mu_0^Y$, $\hat{\mu}_0^A = \mu_0^A$, and $\hat{\tau}_{\mathrm{init}} = \tau$; or (2)~$\hat{\pi} = \pi$ and $\hat{\delta}_A = \delta_A$; or (3)~$\hat{\pi} = \pi$ and $\hat{\tau}_{\mathrm{init}} = \tau$.
\vspace{-0.2cm}
\end{theorem}
The equalities in Theorem~\ref{thrm:robustness} are meant to hold almost surely. Consistency of $\hat{\tau}_{\mathrm{MRIV}}(x)$ is a direct consequence: If either the nuisance estimators in (1), (2), or (3) converge to their oracle estimands, $\hat{\tau}_{\mathrm{MRIV}}(x)$ will converge to the true CATE. As a result, our \frameworkname framework is \emph{multiply robust} in the sense that our estimator, $\hat{\tau}_{\mathrm{MRIV}}(x)$, is consistent in the union of three different model specifications. Importantly, this is different from \emph{doubly robust} estimators which are only consistent in the union of two model specifications \citep{Wang.2018}. Our \frameworkname is directly applicable to RCTs with non-compliance: Here, the treatment assignment is randomized and the propensity score $\pi(x)$ is known. Our \frameworkname framework can be thus adopted by plugging in the known $\pi(x)$ into the pseudo outcome in Eq.~\eqref{eq:pseudo_outcome}. Moreover, $\hat{\tau}_{\mathrm{MRIV}}(x)$ is already consistent if either $\hat{\tau}_{\mathrm{init}}$(x) or $\hat{\delta}_A(x)$ are.

\subsection{Theoretical analysis}
\label{sec:theory}

We derive the asymptotic bound on the convergence rate of \frameworkname under smoothness assumptions. For this, we define $s$-smooth functions as functions contained in the Hölder class $\mathcal{H}(s)$, associated with Stone's minimax rate \citep{Stone.1980} of $n^{-2s/(2s+p)}$, where $p$ is the dimension of $\mathcal{X}$.

\begin{assumption}[Smoothness]\label{ass:smoothness}
We assume that (1)~the nuisance component $\mu_0^Y(\cdot)$ is $\alpha$-smooth, $\mu_0^A(\cdot)$ is $\beta$-smooth, $\pi(\cdot)$ is $\gamma$-smooth, and $\delta_A(\cdot)$ is $\delta$-smooth; (2)~all nuisance components are estimated with their respective minimax rate of $n^{\frac{-2k}{2k+p}}$, where $k \in \{\alpha, \beta, \gamma, \delta\}$; and (3)~the oracle CATE $\tau(\cdot)$ is $\eta$-smooth and the initial CATE estimator $\hat{\tau}_{\mathrm{init}}$ converges with rate $r_{\tau}(n)$. We provide a rigorous definition in Appendix~\ref{app:rates}.
\end{assumption}

Assumption~\ref{ass:smoothness} for smoothness provides us with a way to quantify the difficulty of the underlying nonparametric regression problems. Similar assumptions have been imposed for asymptotic analysis of previous CATE estimators in \citep{Kennedy.2022d, Curth.2021}. They can be replaced with other assumptions such as assumptions on the level of sparsity of the CATE components. We also provide an asymptotic analysis under sparsity assumptions (see Appendix~\ref*{app:sparsity}). 


\begin{assumption}[Boundedness]\label{ass:boundedness}
We assume that there exist constants $C, \rho, \widetilde{\rho}, \epsilon, K>0$ such that for all $x \in \mathcal{X}$ it holds that: (1)~$|\mu_i^Y(x)| \leq C$; (2)~$|\delta_A (x)| = |\mu_1^A(x) - \mu_0^A(x)| \geq \rho$ and $|\hat{\delta}_A (x)| \geq \widetilde{\rho}$; (3)~$\epsilon \leq \hat{\pi}(x) \leq 1 - \epsilon$; and (4)~$|\hat{\tau}_{\mathrm{init}}(x)| \leq K$.
\end{assumption}
Assumptions \ref{ass:boundedness}.1, \ref{ass:boundedness}.3, and \ref{ass:boundedness}.4 are standard and in line with previous works on theoretical analyses of CATE estimators \citep{Curth.2021,Kennedy.2022d}. Assumption \ref{ass:boundedness}.2 ensures that both the oracle CATE and the estimator are bounded. Violations of Assumption \ref{ass:boundedness}.2 may occur when working with ``weak'' instruments, which are IVs that are only weakly correlated with the treatment. Using IV methods with weak instruments should generally be avoided \citep{Li.2022}. However, in many applications such as RCTs with non-compliance, weak instruments are unlikely to occur as patients' compliance decisions are generally correlated with the initial treatment assignments.

We state now our main theoretical result: an upper bound on the oracle risk of the \frameworkname estimator. To derive our bound, we leverage the sample splitting approach from \citep{Kennedy.2022d}. The approach in \citep{Kennedy.2022d} has been initially used to analyze the DR-learner for CATE estimation under unconfoundedness and allows for the derivation of robust convergence rates. It has later been adapted to several other meta learners \citep{Curth.2021}, yet \underline{not} for IV methods.

\begin{theorem}[Oracle upper bound under smoothness]\label{thrm:upperbound}
Let $\mathcal{D}_\ell$ for $\ell \in \{1,2,3\}$ be independent samples of size $n$. Let $\hat{\tau}_{init}(x)$, $\hat{\mu}_0^Y(x)$, and $\hat{\mu}_0^A(x)$ be trained on $\mathcal{D}_1$, and let $\hat{\delta}_A(x)$ and $\hat{\pi}(x)$ be trained on $\mathcal{D}_2$. We denote $\hat{Y}_{\mathrm{\mathrm{MR}}}$ as the pseudo outcome from Eq.~\eqref{eq:pseudo_outcome} and $\hat{\tau}_{\mathrm{MRIV}}(x) = \hat{\E}_n[\hat{Y}_{\mathrm{\mathrm{MR}}} \mid X = x]$ as the pseudo outcome regression on $\mathcal{D}_3$ for some generic estimator $\hat{\E}_n[ \cdot \mid X =x]$ of $\E[ \cdot \mid X =x]$.

We assume that the second-stage estimator $\hat{\E}_n$ yields the minimax rate $n^{-\frac{2\eta}{2\eta + p}}$ and satisfies the stability assumption from \citet{Kennedy.2022d}, Proposition 1 (see Appendix~\ref{app:proofs}). Then, under Assumptions 1-4 the oracle risk is upper bounded by
\begin{equation*}\label{eq:upperbound_mriv}
    \E\left[\left(\hat{\tau}_{\mathrm{MRIV}}(x) - \tau(x)\right)^2\right] \lesssim   n^{\frac{-2\eta}{2\eta+p}} + r_{\tau}(n) \left(n^{\frac{-2\gamma}{2\gamma+p}} + n^{\frac{-2\delta}{2\delta+p}}  \right) + 
    n^{-2\left(\frac{\alpha}{2\alpha+p}+\frac{\gamma}{2\gamma+p}\right)} +
    n^{-2\left(\frac{\beta}{2\beta+p}+\frac{\gamma}{2\gamma+p}\right)} .
\end{equation*}
\end{theorem}
\vspace{-0.4cm}
\begin{proof}
See Appendix~\ref*{app:proofs}. The proof provides a more general bound which depends on the pointwise mean squared errors of the nuisance parameters (Lemma~\ref{lem:general_bound}).
\end{proof}

\vspace{-0.3cm}

Recall that the first summand of the lower bound in Eq.~\eqref{eq:upperbound_mriv} is the minimax rate for the oracle CATE $\tau(x)$ which cannot be improved upon. Hence, for a fast convergence rate of $\hat{\tau}_{\mathrm{MRIV}}(x)$, it is sufficient if either: (1)~$r_{\tau}(n)$ decreases fast and $\alpha$, $\beta$ are large; (2)~$\gamma$ and $\delta$ are large; or (3)~$r_{\tau}(n)$ decreases fast and $\gamma$ is large. This is in line with the multiply robustness property of \frameworkname (Theorem~\ref{thrm:robustness}) and means that \frameworkname achieves a fast rate even if the initial or several nuisance estimators converge slowly.

\textbf{Improvement over $\hat{\tau}_{\mathrm{init}}(x)$:}
From the bound in Eq.~\eqref{eq:upperbound_mriv}, it follows that $\hat{\tau}_{\mathrm{MRIV}}(x)$ improves on the convergence rate of the initial CATE estimator $\hat{\tau}_{\mathrm{init}}(x)$ if its rate $r_{\tau}(n)$ is lower bounded by
\begin{equation}\label{eq:lowerbound_mriv}
    r_{\tau}(n) \gtrsim n^{\frac{-2\eta}{2\eta+p}} + n^{-2 \, \left(\frac{\alpha}{2\alpha+p}+\frac{\gamma}{2\gamma+p}\right)} +
    n^{-2\, \left(\frac{\beta}{2\beta+p}+\frac{\gamma}{2\gamma+p}\right)}.
\end{equation}
Hence, our \frameworkname estimator is more likely to improve on the initial estimator $\hat{\tau}_{\mathrm{init}}(x)$ if either (1)~$\gamma$ is large or (2)~$\alpha$ and $\beta$ are large. Note that the margin of improvement depends also on the size of $\gamma$ and $\delta$, i.e., on the smoothness of $\pi(x)$ and $\delta_A(x)$. In fact, this is widely fulfilled in practice. For example, the former is fulfilled for RCTs with non-compliance, where $\pi(x)$ is often some known, fixed number $p \in (0,1)$. 

\subsection{\frameworkname vs. Wald estimator} 

We compare $\hat{\tau}_{\mathrm{MRIV}}(x)$ to the Wald estimator $\hat{\tau}_{\mathrm{W}}(x)$. First, we derive an asymptotic upper bound.

\begin{theorem}[Wald oracle upper bound]\label{thrm:rate_wald}
Assume that $\mu_1^Y(x)$, $\mu_0^Y(x)$ are $\alpha$-smooth, $\mu_1^A(x)$, $\mu_0^A(x)$ are $\beta$-smooth, and are estimated with their respective minimax rate. Let $\hat{\delta}_A(x) = \hat{\mu}_1^A(x) - \hat{\mu}_0^A(x)$ satisfy Assumption~\ref{ass:boundedness}. Then, the oracle risk of the Wald estimator $\hat{\tau}_W(x)$ is bounded by
\begin{equation}\label{eq:upperbound_wald}
    \E\left[(\hat{\tau}_{\mathrm{W}}(x) - \tau(x))^2\right] \lesssim  n^{-\frac{2\alpha}{2\alpha+p}} + n^{-\frac{2\beta}{2\beta+p}}.
\end{equation}
\end{theorem}
\begin{proof}
See Appendix~\ref*{app:proofs}.
\end{proof}
\vspace{-0.3cm}
We now consider the \frameworkname estimator $\hat{\tau}_{\mathrm{MRIV}}(x)$ with $\hat{\tau}_{\mathrm{init}} = \hat{\tau}_{\mathrm{W}}(x)$, i.e., initialized with the Wald estimator (under sample splitting). Plugging the Wald rate from Eq.~\eqref{eq:upperbound_wald} into the Eq.~\eqref{eq:upperbound_mriv} yields

\vspace{-0.6cm}
\begin{equation*}
    \resizebox{\hsize}{!}{$\E\left[\left(\hat{\tau}_{\mathrm{MRIV}}(x) - \tau(x)\right)^2\right] \lesssim n^{\frac{-2\eta}{2\eta+p}} +
    n^{-2\left(\frac{\alpha}{2\alpha+p}+\frac{\delta}{2\delta+p}\right)} +
    n^{-2\left(\frac{\beta}{2\beta+p}+\frac{\delta}{2\delta+p}\right)} +
    n^{-2\left(\frac{\alpha}{2\alpha+p}+\frac{\gamma}{2\gamma+p}\right)} +
    n^{-2\left(\frac{\beta}{2\beta+p}+\frac{\gamma}{2\gamma+p}\right)}$}.
\end{equation*}

\vspace{-0.5cm}

For $\alpha = \beta = \gamma = \delta$, the rates of $\hat{\tau}_{\mathrm{MRIV}}(x)$ and $\hat{\tau}_{\mathrm{W}}(x)$ reduce to $\E\left[\left(\hat{\tau}_{\mathrm{MRIV}}(x) - \tau(x)\right)^2\right] \lesssim n^{\frac{-2\eta}{2\eta+p}} +
    n^{\frac{-4\alpha}{2\alpha+p}}$ and $\E\left[\left(\hat{\tau}_{\mathrm{W}}(x) - \tau(x)\right)^2\right] \lesssim 
    n^{\frac{-2\alpha}{2\alpha+p}}$.
Hence, $\hat{\tau}_{\mathrm{MRIV}}(x)$ outperforms $\hat{\tau}_{\mathrm{W}}(x)$ asymptotically for $\eta > \alpha$, i.e., when the CATE $\tau(x)$ is smoother than its components, which is usually the case in practice \citep{Kunzel.2019}. For $\eta = \alpha$, the rates of both estimators coincide. Hence, we should expect MRIV to improve on the Wald estimator in real-world settings with large sample size.

\vspace{-0.2cm}
\subsection{\modelname}
\label{sec:neural_nets}

Based on our \frameworkname framwork, we develop a tailored deep neural network called \modelname for CATE estimation using IVs. Our \modelname produces both an initial CATE estimator $\hat{\tau}_{\mathrm{init}}(x)$ and nuisance estimators $\hat{\mu}_0^Y(x)$, $\hat{\mu}_0^A(x)$, $\hat{\delta}_A(x)$, and $\hat{\pi}(x)$. 

\begin{wrapfigure}{r}{0.25\textwidth}
 \begin{center}
 \vspace{-1.1cm}
\includegraphics[width=0.25\textwidth]{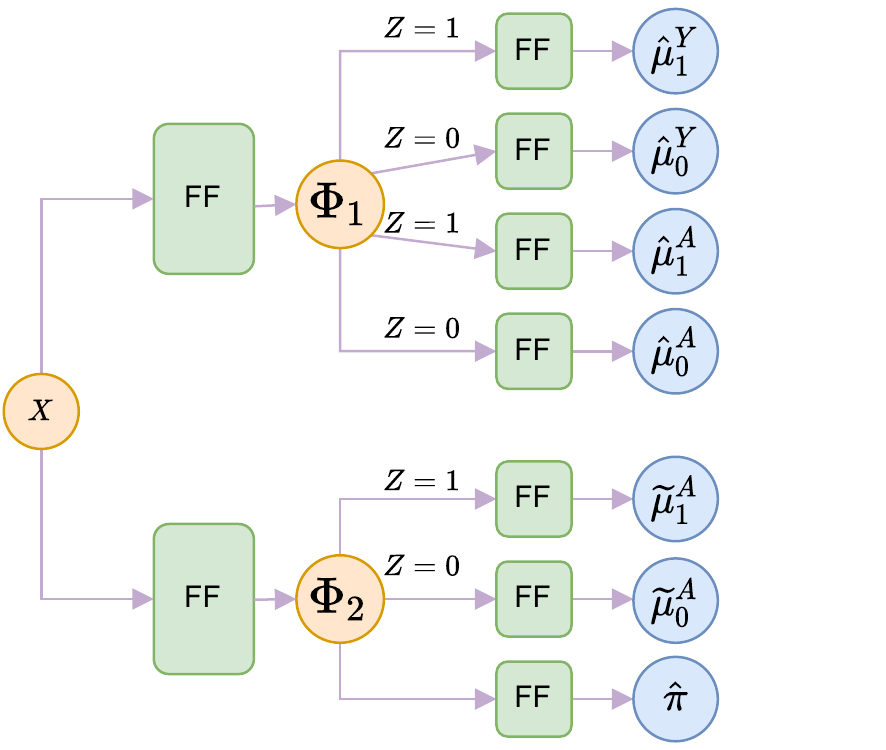}
\end{center}
 \vspace{-0.3cm}
\caption{Architecture of \modelname.}
\label{fig:mrnet}
 \vspace{-0.7cm}
\end{wrapfigure}

For \modelname, we choose deep neural networks for the nuisance components due to their predictive power and their ability to learn complex shared representations for several nuisance components. Sharing representations between nuisance components has been exploited previously for CATE estimation, yet only under unconfoundedness \citep{Shalit.2017, Curth.2021}. Building shared representations is more efficient in finite sample regimes than estimating all nuisance components separately as they usually share some common structure. 

In \modelname, not all nuisance components should share a representation. Recall that, in Theorem~\ref{thrm:upperbound}, we assumed that (1)~$\hat{\tau}_{\mathrm{init}}(x)$, $\hat{\mu}_0^Y(x)$, and $\hat{\mu}_0^A(x)$; and (2)~$\hat{\delta}_A(x)$ and $\hat{\pi}(x)$ are trained on two independent samples in order to derive the upper bound on the oracle risk. Hence, we propose to build two separate representations $\Phi_1$ and $\Phi_2$, so that (i)~$\Phi_1$ is used to learn the parameters (1), 
and (ii)~$\Phi_2$ is used to learn the parameters (2). 
This ensures that the nuisance estimators (1) share minimal information with nuisance estimators (2) even though they are estimated on the same data (cf. \citep{Curth.2021}). 

The architecture of \modelname is shown in Fig.~\ref{fig:mrnet}. \modelname takes the observed covariates $X$ as input to build the two representations $\Phi_1$ and $\Phi_2$. The first representation $\Phi_1$ is used to output estimates $\hat{\mu}_1^Y(x)$, $\hat{\mu}_0^Y(x)$, $\hat{\mu}_1^A(x)$, and $\hat{\mu}_0^A(x)$ of the CATE components. The second representation $\Phi_2$ is used to output estimates $\widetilde{\mu}_1^A(x)$, $\widetilde{\mu}_0^A(x)$, and $\hat{\pi}(x)$. \modelname is trained by minimizing an overall loss
\begin{equation}
    \resizebox{.8\hsize}{!}{$\mathcal{L}(\theta) = \sum_{i=1}^n \left[ \left(\hat{\mu}_{z_i}^Y(x_i) - y_i \right)^2 + \mathrm{BCE}\left(\hat{\mu}_{z_i}^A(x_i), a_i \right) + \mathrm{BCE}\left(\widetilde{\mu}_{z_i}^A(x_i), a_i \right) +  \mathrm{BCE}\left(\hat{\pi}(x_i), z_i \right) \right]$},
\end{equation}
where $\theta$ denotes the neural network parameters and $\mathrm{BCE}$ is the binary cross entropy loss. After training \modelname, we obtain the $\hat{\tau}_{\mathrm{init}}(x) = \frac{\hat{\mu}_1^Y(x) - \hat{\mu}_0^Y(x)}{\hat{\mu}_1^A(x) - \hat{\mu}_0^A(x)}$ and obtain the nuisance estimators $\hat{\mu}_0^Y(x)$, $\hat{\mu}_0^A(x)$, $\hat{\delta}_A(x) = \widetilde{\mu}_1^A(x) - \widetilde{\mu}_0^A(x)$ and $\hat{\pi}(x)$. Then, we perform, we perform the pseudo regression (Stage~2) of \frameworkname to obtain $\hat{\tau}_{\mathrm{\frameworkname}}(x)$.

\textbf{Implementation:} 
Details on the implementation, the network architecture and hyperparameter tuning are in Appendix~\ref*{app:hyper}. We perform both the training of \modelname and the pseudo outcome regression on the full training data. Needless to say, \modelname can be easily adopted for sample splitting or cross-fitting procedures as in \citet{Chernozhukov.2018}, namely, by learning separate networks for each representation $\Phi_1$ and $\Phi_2$. In our experiments, we do not use sample splitting or cross-fitting, as this can affect the performance in finite sample regimes. Of note, our choice is consistent with previous work \citep{Curth.2021}. In Appendix~\ref{app:results_cf} we report results using cross-fitting.

\section{Computational experiments}
\label{sec:experiments}

\subsection{Simulated data}\label{sec:exp_sim}

In causal inference literature, it is common practice to use simulated data for performance evaluations \citep{Bica.2020, Curth.2021, Hartford.2017}. Simulated data offers the crucial benefit that it provides ground-truth information on the counterfactual outcomes and thus allows for direct benchmarking against the oracle CATE.

\textbf{Data generation:}
We generate simulated data by sampling the oracle CATE $\tau(x)$ and the nuisance components $\mu_i^Y(x)$, $\mu_i^A(x)$, and $\pi(x)$ from Gaussian process priors. Using Gaussian processes has the following advantages: (1)~It allows for a fair method comparison, as there is no need to explicitly specify the nuisance components, which could lead to unwanted inductive biases favoring a specific method; (2)~the sampled nuisance components are non-linear and thus resemble real-world scenarios where machine learning methods would be applied; and, (3)~by sampling from the prior induced by the Mat{\'e}rn kernel \citep{Rasmussen.2008}, we can control the smoothness of the nuisance components, which allows us to confirm our theoretical results from Sec.~\ref{sec:theory}.
For a detailed description of our data generating process, we refer to Appendix~\ref*{app:sim}.

\textbf{Baselines:}
We compare our \modelname with state-of-the-art IV baselines. Details regarding baselines and nuisance parameter estimation are in Appendix~\ref*{app:baseline}. Note that many of the baselines do not directly aim at CATE estimation but rather at counterfactual outcome prediction. We nevertheless use these methods as baselines and, for this, obtain the CATE by taking the difference between the predictions of the factual and counterfactual outcomes.

\begin{wraptable}{r}{8cm}
\vspace{-0.8cm}
\scriptsize
\caption{Performance comparison: our \modelname vs. existing baselines.}
\label{tab:base}
\resizebox{8cm}{!}{
\begin{tabular}{lccc}
\toprule
{Method} & {$n = 3000$} & {$n = 5000$} & {$n = 8000$} \\
\midrule
\textsc{(1) Standard ITE} & &  & \\
\quad TARNet \citep{Shalit.2017} &$0.76 \pm 0.14$& $0.70 \pm 0.12$ & $0.69 \pm 0.17$\\
\quad TARNet + DR \citep{Shalit.2017, Kennedy.2022d} &$0.78 \pm 0.10$& $0.66 \pm 0.09$ & $0.70 \pm 0.10$\\
\midrule
\textsc{(2) General IV} & &  & \\
\quad 2SLS \citep{Wooldridge.2013}&$1.22 \pm 0.23$ &$0.79 \pm 0.37$ & $1.12 \pm 0.29$ \\
\quad KIV \citep{Singh.2019}&$1.54 \pm 0.53$ & $1.18 \pm 1.14$ & $3.80 \pm 4.71$\\
\quad DFIV \citep{Xu.2021}&$0.43 \pm 0.11$ & $0.40 \pm 0.21$ & $0.46 \pm 0.54$ \\
\quad DeepIV \citep{Hartford.2017}&$0.96 \pm 0.30$ & $0.28 \pm 0.09$ & $0.23 \pm 0.04$\\
\quad DeepGMM \citep{Bennett.2019}&$0.95 \pm 0.38$ &$0.37 \pm 0.09$ & $0.42 \pm 0.14$ \\
\quad DMLIV \citep{Syrgkanis.2019}&$1.92 \pm 0.71$ & $0.92 \pm 0.41$ & $1.14 \pm 0.24$\\
\quad DMLIV + DRIV \citep{Syrgkanis.2019}&$0.41 \pm 0.12$ & $0.22 \pm 0.04$ & $0.21 \pm 0.06$\\
\midrule
\textsc{(3) Wald estimator \citep{Wald.1940}} & & & \\
\quad Linear &$1.06 \pm 0.63$ & $0.62 \pm 0.22$ & $0.81 \pm 0.34$\\
\quad BART &$0.95 \pm 0.30$ &$0.63 \pm 0.33$ &$0.88 \pm 0.28$\\ \bottomrule \noalign{\smallskip}
\modelname (ours) &$\boldsymbol{0.26 \pm 0.11}$ & $\boldsymbol{0.15 \pm 0.03}$ & $\boldsymbol{0.13 \pm 0.03}$\\

\bottomrule
\multicolumn{4}{l}{Reported: RMSE for base methods (mean $\pm$ standard deviation). Lower $=$ better (best in bold)}
\end{tabular}
}
\vspace{-0.5cm}
\end{wraptable}

\textbf{Performance evaluation:} For all experiments, we use a 80/20 split as training/test set. We calcalute the root mean squared errors (RMSE) between the CATE estimates and the oracle CATE on the test set. We report the mean RMSE and the standard deviation over five data sets generated from random seeds.

\textbf{Results:}
Table~\ref{tab:base} shows the results for all baselines. Here, the DR-learner does not improve the performance of TARNet, which is reasonable as both the DR-learner and TARNet assume unconfoundedness and are thus biased in our setting. Our \modelname outperforms all baselines. Our \modelname also achieves a smaller standard deviation. For additional results, we refer to Appendix~\ref*{app:results}.

We further compare the performance of two different meta-learner frameworks -- DRIV \citep{Syrgkanis.2019} and our \frameworkname -- across different base methods. The results are in Table~\ref{tab:frameworks}.
 \begin{table}[h]
 \vspace{-0.6cm}
\caption{Base model with different meta-learners (i.e., none, DRIV, and our \frameworkname).}
\label{tab:frameworks}
\centering
\resizebox{\columnwidth}{!}{%
\begin{tabular}{lccccccccc}
\noalign{\smallskip} \toprule \noalign{\smallskip}
& \multicolumn{3}{c}{$n = 3000$} & \multicolumn{3}{c}{$n = 5000$} & \multicolumn{3}{c}{$n = 8000$} \\
\cmidrule(lr){2-4} \cmidrule(lr){5-7} \cmidrule(lr){8-10}
\backslashbox{Base methods}{Meta-learners} & None & DRIV & \frameworkname (ours) & None & DRIV & \frameworkname (ours) & None & DRIV & \frameworkname (ours)\\
\midrule
\textsc{(1) Standard ITE} & & & \\
\quad TARNet \citep{Shalit.2017}& $0.76 \pm 0.14$&$\boldsymbol{0.31 \pm 0.05}$& $0.34 \pm 0.13$ &$0.70 \pm 0.12$ & $\boldsymbol{0.17 \pm 0.06}$ & $\boldsymbol{0.17 \pm 0.05}$ &$0.69 \pm 0.17$ & $0.21 \pm 0.04$ & $\boldsymbol{0.16 \pm 0.04}$\\
\midrule
\textsc{(2) General IV} & & & \\
\quad 2SLS \citep{Wooldridge.2013}& $1.22 \pm 0.23$& $0.40 \pm 0.11$ &$\boldsymbol{0.31 \pm 0.08}$ &$0.79 \pm 0.37$ & $\boldsymbol{0.17 \pm 0.09}$ & $0.19 \pm 0.05$ & $1.12 \pm 0.29$ & $0.21 \pm 0.05$& $\boldsymbol{0.16 \pm 0.02}$\\
\quad KIV \citep{Singh.2019}&$1.54 \pm 0.53$ & $0.40 \pm 0.10$& $\boldsymbol{0.39 \pm 0.11}$ & $1.18 \pm 1.14$& $0.20 \pm 0.08$ & $\boldsymbol{0.17 \pm 0.06}$ & $3.80 \pm 4.71$& $0.31 \pm 0.18$& $\boldsymbol{0.28 \pm 0.19}$\\
\quad DFIV \citep{Xu.2021}&$0.43 \pm 0.11$ &$\boldsymbol{0.26 \pm 0.05}$ & $0.27 \pm 0.07$ &$0.40 \pm 0.21$ & $0.18 \pm 0.09$ & $\boldsymbol{0.16 \pm 0.04}$ &$0.46 \pm 0.54$ & $0.21 \pm 0.06$ & $\boldsymbol{0.18 \pm 0.05}$ \\
\quad DeepIV \citep{Hartford.2017}&$0.96 \pm 0.30$ & $0.27 \pm 0.03$ & $\boldsymbol{0.26 \pm 0.05}$ & $0.28 \pm 0.09$& $\boldsymbol{0.18 \pm 0.08}$ & $\boldsymbol{0.18 \pm 0.05}$ & $0.23 \pm 0.04$ & $0.21 \pm 0.03$ & $\boldsymbol{0.16 \pm 0.03}$\\
\quad DeepGMM \citep{Bennett.2019}&$0.95 \pm 0.38$ & $0.40 \pm 0.15$& $\boldsymbol{0.36 \pm 0.13}$ &$0.37 \pm 0.09$ & $0.24 \pm 0.12$ & $\boldsymbol{0.16 \pm 0.05}$ &$0.42 \pm 0.14$ & $0.21 \pm 0.03$ & $\boldsymbol{0.17 \pm 0.03}$ \\
\quad DMLIV \citep{Syrgkanis.2019}&$1.92 \pm 0.71$ &$0.41 \pm 0.12$ & $\boldsymbol{0.37 \pm 0.11}$ &$0.92 \pm 0.41$ & $0.22 \pm 0.05$ & $\boldsymbol{0.16 \pm 0.05}$ &$1.14 \pm 0.24$ & $0.21 \pm 0.06$& $\boldsymbol{0.18 \pm 0.05}$\\
\midrule
\textsc{(3) Wald estimator \citep{Wald.1940}}& & & & &  & \\
\quad Linear &$1.06 \pm 0.63$ & $0.42 \pm 0.15$& $\boldsymbol{0.38 \pm 0.14}$ &$0.62 \pm 0.22$ &  $\boldsymbol{0.19 \pm 0.09}$ & $0.25 \pm 0.09$ &$0.81 \pm 0.34$ & $0.19 \pm 0.09$ & $\boldsymbol{0.18 \pm 0.04}$\\
\quad BART &$0.95 \pm 0.30$ &$0.48 \pm 0.14$ & $\boldsymbol{0.46 \pm 0.12}$ &$0.63 \pm 0.33$ & $0.26 \pm 0.13$ & $\boldsymbol{0.20 \pm 0.07}$ &$0.88 \pm 0.28$ & $0.31 \pm 0.08$ & $\boldsymbol{0.29 \pm 0.04}$\\
\midrule
\modelname {\textbackslash}w network only (ours) &$0.39 \pm 0.13$ &$0.35 \pm 0.12$ & $\boldsymbol{0.26 \pm 0.11}$ & $0.31 \pm 0.04$& $0.19 \pm 0.13$ & $\boldsymbol{0.15 \pm 0.03}$ & $0.26 \pm 0.06$& $0.18 \pm 0.08$&  $\boldsymbol{0.13 \pm 0.03}$\\

\bottomrule
\multicolumn{10}{l}{Reported: RMSE (mean $\pm$ standard deviation). Lower $=$ better (best improvement over none meta-learner in bold)}
\end{tabular}%
}
\vspace{-0.5cm}
\end{table}
The nuisance parameters are estimated using feed forward neural networks (DRIV) or TARNets with either binary or continuous outputs (MRIV). Our \frameworkname improves over the variant without any meta-learner framework across all base methods (both in terms of RMSE and standard deviation). Furthermore, \frameworkname is clearly superior over DRIV. This demonstrates the effectiveness of our \frameworkname across different base methods (note: \frameworkname with an arbitrary base model is typically superior to DRIV with our custom network from above). 
\begin{wraptable}{r}{8cm}
\vspace{-0.8cm}
\caption{Ablation study.}
\label{tab:ablation}
\centering
\scriptsize
\resizebox{8cm}{!}{
\begin{tabular}{lccc}
\noalign{\smallskip} \toprule \noalign{\smallskip}
{Method} & {$n = 3000$} & {$n = 5000$} & {$n = 8000$} \\
\midrule
\modelname {\textbackslash}w network only &$0.39 \pm 0.13$ & $0.31 \pm 0.04$ &
$0.26 \pm 0.06$\\
\modelname {\textbackslash}w single repr. &$0.28 \pm 0.12$ & $0.21 \pm 0.04$ &
$0.32 \pm 0.10$\\
\modelname (ours) &$\boldsymbol{0.26 \pm 0.11}$ & $\boldsymbol{0.15 \pm 0.03}$ &
$\boldsymbol{0.13 \pm 0.03}$\\
\bottomrule
\multicolumn{4}{l}{Reported: RMSE (mean $\pm$ standard deviation). Lower $=$ better (best in bold)}
\end{tabular}}
\vspace{-0.8cm}
\end{wraptable}
\modelname is overall best. We also performed additional experiments where we used semi-synthetic data and cross-fitting approaches for both meta-learners (see Appendix~\ref{app:results} and \ref{app:results_cf}).

\textbf{Ablation study:} Table~\ref{tab:ablation} compares different variants of our \modelname.
These are: (1)~\frameworkname but network only; (2)~\modelname with a single representation for all nuisance estimators; and (3)~our \modelname from above. We observe that \modelname is best. This justifies our proposed network architecture for \modelname. Hence, combing the result from above, our performance gain must be attributed to both our framework \underline{and} the architecture of our deep neural network. 

\begin{wrapfigure}{r}{0.68\textwidth}
\vspace{-0.5cm}
\centering
\includegraphics[width=0.7\textwidth]{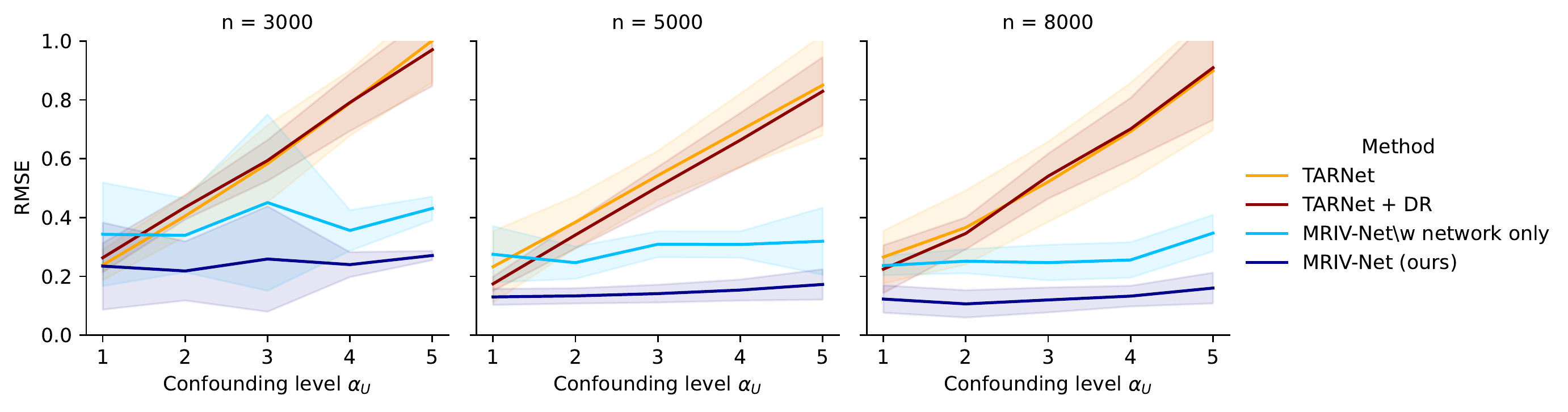}
\vspace{-0.8cm}
\caption{Results over different levels of confounding $\alpha_U$. Shaded area shows standard deviation.}
\label{fig:plt_confounding}
\vspace{-0.4cm}
\end{wrapfigure}

\textbf{Robustness checks for unobserved confounding and smoothness:} Here, we demonstrate the importance of handling unobserved confounding (as we do in our \frameworkname framework). For this, Fig.~\ref{fig:plt_confounding} plots the results for our \modelname vs. standard CATE without customization for confounding (i.e., TARNet with and without the DR-learner) over over different levels of unobserved confounding. The RMSE of both TARNet variants increase almost linearly with increasing confounding. In contrast, the RMSE of our \modelname only marginally. Even for low confounding regimes, our \modelname performs competitively. 
\begin{wrapfigure}{r}{0.35\textwidth}
\vspace{-0.8cm}
 \begin{center}
\includegraphics[width=0.35\textwidth]{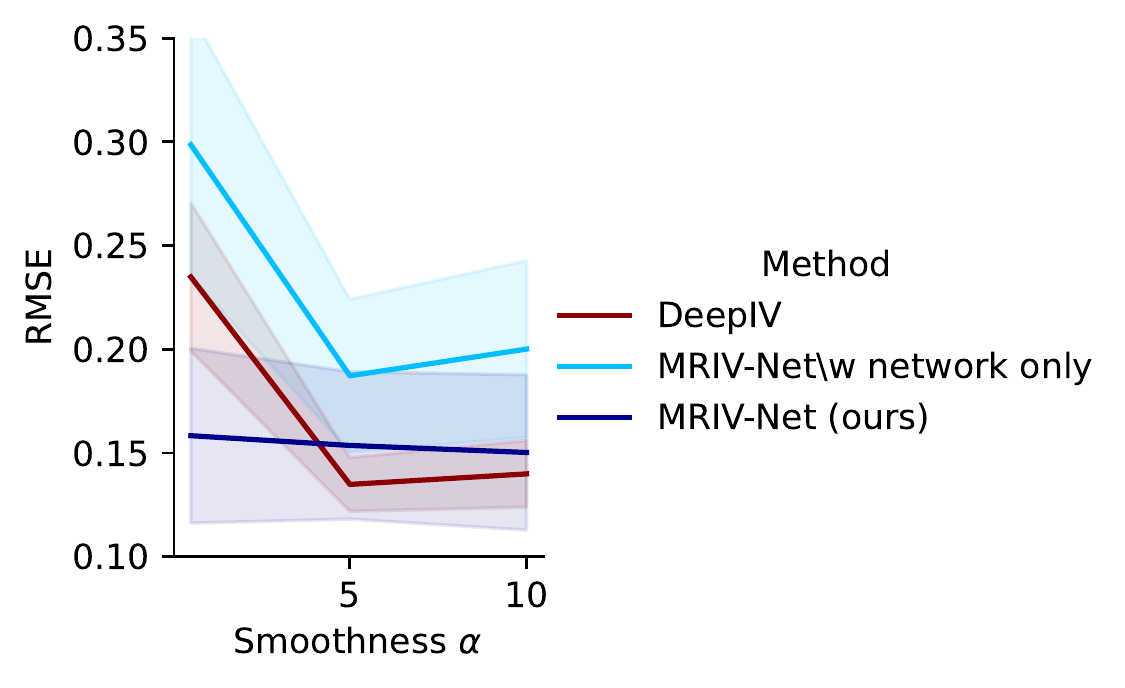}
\end{center}
\vspace{-0.3cm}
\caption{Results over different levels of smoothness $\alpha$ of $\mu_i^Y(\cdot)$, sample size $n=8000$. Larger $\alpha$ = smoother. Shaded areas show standard deviation.}
\vspace{-1cm}
\label{fig:plt_smoothness}
\end{wrapfigure}

Fig.~\ref{fig:plt_smoothness} varies the smoothness level. This is given by $\alpha$ of $\mu_i^Y(\cdot)$ (controlled by the Matérn kernel prior).
Here, the performance decreases for the baselines, i.e., DeepIV and our network without \frameworkname framework. In contrast, the peformance of our \modelname remains robust and outperforms the baselines. This confirms our theoretical results from above. It thus indicates that our \frameworkname framework works best when the oracle CATE $\tau(x)$ is smoother than the nuisance parameters $\mu_i^Y(x)$.

\subsection{Case study with real-world data}
\label{sec:exp_real}
\vspace{-0.2cm}
\textbf{Setting:} We demonstrate effectiveness of our framework using a case study with real-world, medical data. Here, we use medical data from the so-called \emph{Oregon health insurance experiment} (OHIE) \citep{Finkelstein.2012}. It provides data for an RCT with non-compliance: In 2008, $\sim$30,000 low-income, uninsured adults in Oregon were offered participation in a health insurance program by a lottery. Individuals whose names were drawn could decide to sign up for health insurance. After a period of 12 months, in-person interviews took place to evaluate the health condition of the respective participant. 

\begin{wrapfigure}{l}{0.5\textwidth}
\vspace{-0.6cm}
 \begin{center}
\includegraphics[width=0.5\textwidth]{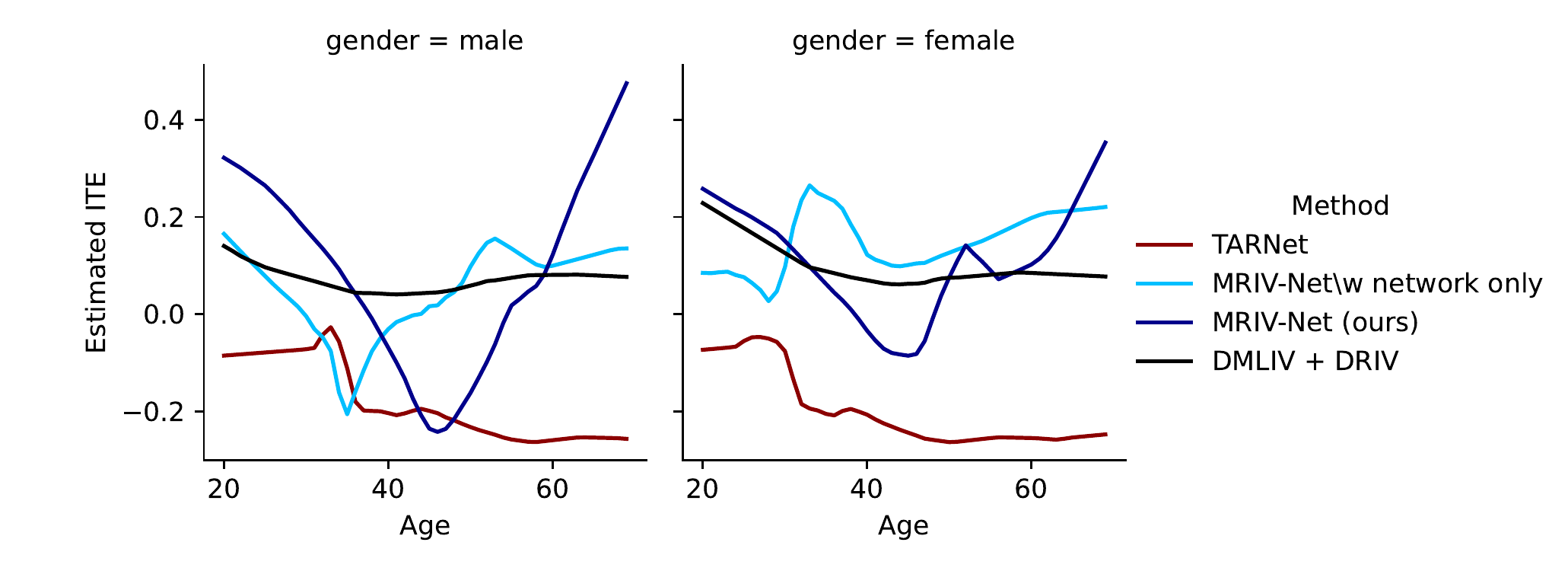}
\end{center}
\vspace{-0.5cm}
\caption{Results on real-world medical data.}
\vspace{-0.5cm}
\label{fig:real}
\end{wrapfigure}
\vspace{-0.3cm}
In our analysis, the lottery assignment is the instrument $Z$, the decision to sign up for health insurance is treatment $A$, and an overall health score is the outcome $Y$. We also include five covariates $X$, including age and gender. For details, we refer to Appendix~\ref*{app:real}. We first estimate the CATE function and then report the treatment effect heterogeneity w.r.t. age and gender, while fixing the other covariates. The results for \modelname, our neural network architecture without the \frameworkname framework, and TARNet are in Fig.~\ref{fig:real}.

\textbf{Results:} Our \modelname estimates larger causal effects for an older age. In contrast, TARNet does not estimate positive CATEs even for an older age. 
Even though we cannot evaluate the estimation quality on real-world data, our estimates seem reasonable in light of the medical literature: the benefit of health insurance should increase with older age. This showcases that TARNet may suffer from bias induced by unobserved confounders. We also report the results for DRIV with DMLIV as base method, and observe that in contrast to \modelname, the corresponding CATE does not vary much between ages. Interestingly, both our \modelname estimate a somewhat smaller CATE for middle ages (around 30--50 yrs).
In sum, the findings from our case study are of direct relevance for decision-makers in public health \citep{Imbens.1994}, and highlight the practical value of our framework. We performed further experiments on real-world data which are reported in Appendix~\ref{app:real_furtherexp}.


\clearpage

\paragraph{Reproducibility:} The codes for reproducing the experimental results can be found at \url{https://github.com/DennisFrauen/MRIV-Net}.

\bibliography{bibliography.bib}

\begin{thebibliography}{58}
\providecommand{\natexlab}[1]{#1}
\providecommand{\url}[1]{\texttt{#1}}
\expandafter\ifx\csname urlstyle\endcsname\relax
  \providecommand{\doi}[1]{doi: #1}\else
  \providecommand{\doi}{doi: \begingroup \urlstyle{rm}\Url}\fi

\bibitem[Alaa \& {van der Schaar}(2017)Alaa and {van der Schaar}]{Alaa.2017}
Ahmed~M. Alaa and Mihaela {van der Schaar}.
\newblock Bayesian inference of individualized treatment effects using
  multi-task {G}aussian processes.
\newblock In \emph{NeurIPS}, 2017.

\bibitem[Angrist(1990)]{Angrist.1990}
Joshua~D. Angrist.
\newblock Lifetime earnings and the vietnam era draft lotter: Evidence from
  social security administrative records.
\newblock \emph{The American Economic Review}, 80\penalty0 (3):\penalty0
  313--336, 1990.

\bibitem[Angrist \& Krueger(1991)Angrist and Krueger]{Angrist.1991}
Joshua~D. Angrist and Alan~B. Krueger.
\newblock Does compulsory school attendance affect schooling and earnings?
\newblock \emph{The Quarterly Journal of Economics}, 106\penalty0 (4):\penalty0
  979--1014, 1991.

\bibitem[Angrist et~al.(1996)Angrist, Imbens, and Rubin]{Angrist.1996}
Joshua~D. Angrist, Guido~W. Imbens, and Donald~B. Rubin.
\newblock Identification of causal effects using instrumental variables.
\newblock \emph{Journal of the American Statistical Association}, 91\penalty0
  (434):\penalty0 444--455, 1996.

\bibitem[Bargagli-Stoffi et~al.(2021)Bargagli-Stoffi, de~Witte, and
  Gnecco]{BargagliStoffi.2021}
Falco~J. Bargagli-Stoffi, Kristof de~Witte, and Giorgio Gnecco.
\newblock Heterogeneous causal eﬀects with imperfect compliance: A bayesian
  machine learning approach.
\newblock \emph{Annals of Applied Statistics}, 2021.

\bibitem[Bennett et~al.(2019)Bennett, Kallus, and Schnabel]{Bennett.2019}
Andrew Bennett, Nathan Kallus, and Tobias Schnabel.
\newblock Deep generalized method of moments for instrumental variable
  analysis.
\newblock In \emph{NeurIPS}, 2019.

\bibitem[Bica et~al.(2020{\natexlab{a}})Bica, Alaa, Jordon, and {van der
  Schaar}]{Bica.2020}
Ioana Bica, Ahmed~M. Alaa, James Jordon, and Mihaela {van der Schaar}.
\newblock Estimating counterfactual treatment outcomes over time through
  adversarially balanced representations.
\newblock In \emph{ICLR}, 2020{\natexlab{a}}.

\bibitem[Bica et~al.(2020{\natexlab{b}})Bica, Alaa, and {van der
  Schaar}]{Bica.2020b}
Ioana Bica, Ahmed~M. Alaa, and Mihaela {van der Schaar}.
\newblock Time series deconfounder: Estimating treatment effects over time in
  the presence of hidden confounders.
\newblock In \emph{ICML}, 2020{\natexlab{b}}.

\bibitem[Bloom et~al.(1997)Bloom, Orr, Bell, Cave, Doolittle, Lin, and
  Bos]{Bloom.1997}
Howard~S. Bloom, Larry~L. Orr, Stephen~H. Bell, George Cave, Fred Doolittle,
  Winston Lin, and Johannes~M. Bos.
\newblock The benefits and costs of {JTPA} title ii-a programs: Key findings
  from the national job training partnership act study.
\newblock \emph{Journal of Human Resources}, 32\penalty0 (32):\penalty0
  549--586, 1997.

\bibitem[Chernozhukov et~al.(2018)Chernozhukov, Chetverikov, Demirer, Duflo,
  Hansen, Newey, and Robins]{Chernozhukov.2018}
Victor Chernozhukov, Denis Chetverikov, Mert Demirer, Esther Duflo, Christian
  Hansen, Whitney Newey, and James~M. Robins.
\newblock Double/debiased machine learning for treatment and structural
  parameters.
\newblock \emph{The Econometrics Journal}, 21\penalty0 (1):\penalty0 C1--C68,
  2018.

\bibitem[Chesson(1976)]{Chesson.1976}
Jean Chesson.
\newblock A non-central multivariate hypergeometric distribution arising from
  biased sampling with application to selective predation.
\newblock \emph{Journal of Applied Probability}, 13\penalty0 (4):\penalty0
  795--797, 1976.

\bibitem[Chipman et~al.(2010)Chipman, George, and McCulloch]{Chipman.2010}
Hugh~A. Chipman, Edward~I. George, and Robert~E. McCulloch.
\newblock {BART}: Bayesian additive regression trees.
\newblock \emph{The Annals of Applied Statistics}, 4\penalty0 (1):\penalty0
  266--298, 2010.

\bibitem[Cui \& Tchetgen(2021)Cui and Tchetgen]{Cui.2021}
Yifan Cui and Eric~Tchetgen Tchetgen.
\newblock A semiparametric instrumental variable approach to optimal treatment
  regimes under endogeneity.
\newblock \emph{Journal of the American Statistical Association}, 116\penalty0
  (553):\penalty0 126--137, 2021.

\bibitem[Cui et~al.(2020)Cui, Pu, Shi, Miao, and Tchetgen]{Cui.2020}
Yifan Cui, Hongming Pu, Xu~Shi, Wang Miao, and Eric~Tchetgen Tchetgen.
\newblock Semiparametric proximal causal inference.
\newblock \emph{arXiv preprint}, 2020.

\bibitem[Curth \& {van der Schaar}(2021)Curth and {van der Schaar}]{Curth.2021}
Alicia Curth and Mihaela {van der Schaar}.
\newblock Nonparametric estimation of heterogeneous treatment effects: From
  theory to learning algorithms.
\newblock In \emph{AISTATS}, 2021.

\bibitem[Curth et~al.(2020)Curth, Alaa, and {van der Schaar}]{Curth.2020}
Alicia Curth, Ahmed~M. Alaa, and Mihaela {van der Schaar}.
\newblock Estimating structural target functions using machine learning and
  influence functions.
\newblock \emph{arXiv preprint}, arXiv:2008.06461, 2020.

\bibitem[Daskalakis et~al.(2018)Daskalakis, Ilyas, Syrgkanis, and
  Zeng]{Daskalakis.2018}
Constantinos Daskalakis, Andrew Ilyas, Vasilis Syrgkanis, and Haoyang Zeng.
\newblock Training {GAN}s with optimism.
\newblock In \emph{ICLR}, 2018.

\bibitem[Finkelstein et~al.(2012)Finkelstein, Taubman, Wright, Bernstein,
  Gruber, Newhouse, Allen, and Baicker]{Finkelstein.2012}
Amy Finkelstein, Sarah Taubman, Bill Wright, Mira Bernstein, Jonathan Gruber,
  Joseph~P. Newhouse, Heidi Allen, and Katherine Baicker.
\newblock The oregon health insurance experiment: Evidence from the first year.
\newblock \emph{The Quarterly Journal of Economics}, 127\penalty0 (3):\penalty0
  1057--1106, 2012.

\bibitem[Frauen et~al.(2023)Frauen, Hatt, Melnychuk, and
  Feuerriegel]{Frauen.2022b}
Dennis Frauen, Tobias Hatt, Valentyn Melnychuk, and Stefan Feuerriegel.
\newblock Estimating average causal effects from patient trajectories.
\newblock In \emph{AAAI}, 2023.

\bibitem[Hartford et~al.(2017)Hartford, Lewis, Leyton-Brown, and
  Taddy]{Hartford.2017}
Jason Hartford, Greg Lewis, Kevin Leyton-Brown, and Matt Taddy.
\newblock Deep {IV}: A flexible approach for counterfactual prediction.
\newblock In \emph{ICML}, 2017.

\bibitem[Imbens \& Angrist(1994)Imbens and Angrist]{Imbens.1994}
Guido~W. Imbens and Joshua~D. Angrist.
\newblock Identification and estimation of local average treatment effects.
\newblock \emph{Econometrica}, 62\penalty0 (2):\penalty0 467--475, 1994.

\bibitem[Jesson et~al.(2021)Jesson, Mindermann, Gal, and Shalit]{Jesson.2021}
Andrew Jesson, S{\"o}ren Mindermann, Yarin Gal, and Uri Shalit.
\newblock Quantifying ignorance in individual-level causal-effect estimates
  under hidden confounding.
\newblock In \emph{ICML}, 2021.

\bibitem[Kallus et~al.(2019)Kallus, Mao, and Zhou]{Kallus.2019}
Nathan Kallus, Xiaojie Mao, and Angela Zhou.
\newblock Interval estimation of individual-level causal effects under
  unobserved confounding.
\newblock In \emph{AISTATS}, 2019.

\bibitem[Kennedy(2022)]{Kennedy.2022d}
Edward~H. Kennedy.
\newblock Towards optimal doubly robust estimation of heterogeneous causal
  effects.
\newblock \emph{arXiv preprint}, 2022.

\bibitem[Kennedy et~al.(2019)Kennedy, Lorch, and Small]{Kennedy.2019}
Edward~H. Kennedy, Scott~A. Lorch, and Dylan~S. Small.
\newblock Robust causal inference with continuous instruments using the local
  instrumental variable curve.
\newblock \emph{Journal of the Royal Statistical Society: Series B},
  81\penalty0 (1):\penalty0 121--143, 2019.

\bibitem[Kingma \& Ba(2015)Kingma and Ba]{Kingma.2015}
Diederik~P. Kingma and Jimmy Ba.
\newblock Adam: A method for stochastic optimization.
\newblock In \emph{ICLR}, 2015.

\bibitem[K{\"u}nzel et~al.(2019)K{\"u}nzel, Sekhon, Bickel, and
  Yu]{Kunzel.2019}
S{\"o}ren~R. K{\"u}nzel, Jasjeet~S. Sekhon, Peter~J. Bickel, and Bin Yu.
\newblock Metalearners for estimating heterogeneous treatment effects using
  machine learning.
\newblock \emph{Proceedings of the National Academy of Sciences (PNAS)},
  116\penalty0 (10):\penalty0 4156--4165, 2019.

\bibitem[Li et~al.(2022)Li, Rudin, and McCormick]{Li.2022}
Chunxiao Li, Cynthia Rudin, and Typer~H. McCormick.
\newblock Rethinking nonlinear instrumental variable models through prediction
  validity.
\newblock \emph{Journal of Machine Learning Research}, 23:\penalty0 1--55,
  2022.

\bibitem[Lim et~al.(2018)Lim, Alaa, and {van der Schaar}]{Lim.2018}
Bryan Lim, Ahmed~M. Alaa, and Mihaela {van der Schaar}.
\newblock Forecasting treatment responses over time using recurrent marginal
  structural networks.
\newblock In \emph{NeurIPS}, 2018.

\bibitem[Melnychuk et~al.(2022{\natexlab{a}})Melnychuk, Frauen, and
  Feuerriegel]{Melnychuk.2022}
Valentyn Melnychuk, Dennis Frauen, and Stefan Feuerriegel.
\newblock Causal transformer for estimating counterfactual outcomes.
\newblock In \emph{ICML}, 2022{\natexlab{a}}.

\bibitem[Melnychuk et~al.(2022{\natexlab{b}})Melnychuk, Frauen, and
  Feuerriegel]{Melnychuk.2022b}
Valentyn Melnychuk, Dennis Frauen, and Stefan Feuerriegel.
\newblock Normalizing flows for interventional density estimation.
\newblock \emph{arXiv preprint}, arXiv:2209.06203, 2022{\natexlab{b}}.

\bibitem[Newey \& Powell(2003)Newey and Powell]{Newey.2003}
Whitney~K. Newey and James~L. Powell.
\newblock Instrumental variable estimation of nonparametric models.
\newblock \emph{Econometrica}, 71\penalty0 (5):\penalty0 1565--1578, 2003.

\bibitem[Ogburn et~al.(2015)Ogburn, Rotnitzky, and Robins]{Ogburn.2015}
Elizabeth~L. Ogburn, Andrea Rotnitzky, and James~M. Robins.
\newblock Doubly robust estimation of the local average treatment effect curve.
\newblock \emph{Journal of the Royal Statistical Society: Series B},
  77\penalty0 (2):\penalty0 373--396, 2015.

\bibitem[Okui et~al.(2012)Okui, Small, Tan, and Robins]{Okui.2012}
Ryo Okui, Dylan~S. Small, Zhiqiang Tan, and James~M. Robins.
\newblock Doubly robust instrumental variable regression.
\newblock \emph{Statistica Sinica}, 22\penalty0 (1):\penalty0 173--205, 2012.

\bibitem[Pearl(2009)]{Pearl.2009}
Judea Pearl.
\newblock \emph{Causality}.
\newblock {Cambridge University Press}, New York City, 2009.
\newblock ISBN 9780521895606.

\bibitem[Rasmussen \& Williams(2008)Rasmussen and Williams]{Rasmussen.2008}
Carl~Edward Rasmussen and Christopher K.~I. Williams.
\newblock \emph{Gaussian processes for machine learning}.
\newblock Adaptive computation and machine learning. {MIT Press}, Cambridge,
  Mass., 3. print edition, 2008.
\newblock ISBN 9780262182539.

\bibitem[Robins et~al.(2000)Robins, Hern{\'a}n, and Brumback]{Robins.2000}
James~M. Robins, Miguel~A. Hern{\'a}n, and Babette Brumback.
\newblock Marginal structural models and causal inference in epidemiology.
\newblock \emph{Epidemiology}, 11\penalty0 (5):\penalty0 550--560, 2000.

\bibitem[Rubin(1974)]{Rubin.1974}
Donald~B. Rubin.
\newblock Estimating causal effects of treatments in randomized and
  nonrandomized studies.
\newblock \emph{Journal of Educational Psychology}, 66\penalty0 (5):\penalty0
  688--701, 1974.

\bibitem[Semenova \& Chernozhukov(2021)Semenova and
  Chernozhukov]{Semenova.2021}
Vira Semenova and Victor Chernozhukov.
\newblock Debiased machine learning of conditional average treatment effects
  and other causal functions.
\newblock \emph{The Econometrics Journal}, 24\penalty0 (2):\penalty0 264--289,
  2021.

\bibitem[Shalit et~al.(2017)Shalit, Johansson, and Sontag]{Shalit.2017}
Uri Shalit, Fredrik~D. Johansson, and David Sontag.
\newblock Estimating individual treatment effect: Generalization bounds and
  algorithms.
\newblock In \emph{ICML}, 2017.

\bibitem[Singh \& Sun(2019)Singh and Sun]{Singh.2019b}
Rahul Singh and Liyang Sun.
\newblock Double robustness for complier parameters and a semiparametric test
  for complier characteristics.
\newblock \emph{arXiv preprint}, 2019.

\bibitem[Singh et~al.(2019)Singh, Sahani, and Gretton]{Singh.2019}
Rahul Singh, Maneesh Sahani, and Arthur Gretton.
\newblock Kernel instrumental variable regression.
\newblock In \emph{NeurIPS}, 2019.

\bibitem[Stone(1980)]{Stone.1980}
Charles~J. Stone.
\newblock Optimal rates of convergence for nonparametric estimators.
\newblock \emph{Annals of Statistics}, 8\penalty0 (6), 1980.

\bibitem[Syrgkanis et~al.(2019)Syrgkanis, Lei, Oprescu, Hei, Battocchi, and
  Lewis]{Syrgkanis.2019}
Vasilis Syrgkanis, Victor Lei, Miruna Oprescu, Maggie Hei, Keith Battocchi, and
  Greg Lewis.
\newblock Machine learning estimation of heterogeneous treatment effects with
  instruments.
\newblock In \emph{NeurIPS}, 2019.

\bibitem[Varian(2016)]{Varian.2016}
Hal~R. Varian.
\newblock Causal inference in economics and marketing.
\newblock \emph{Proceedings of the National Academy of Sciences (PNAS)},
  113\penalty0 (27):\penalty0 7310--7315, 2016.

\bibitem[Wager \& Athey(2018)Wager and Athey]{Wager.2018}
Stefan Wager and Susan Athey.
\newblock Estimation and inference of heterogeneous treatment effects using
  random forests.
\newblock \emph{Journal of the American Statistical Association}, 113\penalty0
  (523):\penalty0 1228--1242, 2018.

\bibitem[Wald(1940)]{Wald.1940}
Abraham Wald.
\newblock The fitting of straight lines if both variables are subject to error.
\newblock \emph{Annals of Mathematical Statistics}, 11\penalty0 (3):\penalty0
  284--300, 1940.

\bibitem[Wang et~al.(2021)Wang, Li, and Hopp]{Wang.2021}
Guihua Wang, Jun Li, and Wallace~J. Hopp.
\newblock An instrumental variable forest approach for detecting heterogeneous
  treatment effects in observational studies.
\newblock \emph{Management Science}, 2021.

\bibitem[Wang \& {Tchetgen Tchetgen}(2018)Wang and {Tchetgen
  Tchetgen}]{Wang.2018}
Linbo Wang and Eric~J. {Tchetgen Tchetgen}.
\newblock Bounded, efficient and multiply robust estimation of average
  treatment effects using instrumental variables.
\newblock \emph{Journal of the Royal Statistical Society: Series B},
  80\penalty0 (3):\penalty0 531--550, 2018.

\bibitem[Wang \& Blei(2019)Wang and Blei]{Wang.2019}
Yixin Wang and David~M. Blei.
\newblock The blessings of multiple causes.
\newblock \emph{Journal of the American Statistical Association}, 114\penalty0
  (528):\penalty0 1574--1596, 2019.

\bibitem[Wooldridge(2013)]{Wooldridge.2013}
Jeffrey~M. Wooldridge.
\newblock \emph{Introductory Econometrics: A modern approach}.
\newblock Routledge, 2013.
\newblock ISBN 9781136586101.

\bibitem[Wright(1928)]{Wright.1928}
Phillip~G. Wright.
\newblock \emph{The tariff on animal and vegitable oils}.
\newblock Macmillan, New York, 1928.

\bibitem[Xu et~al.(2021{\natexlab{a}})Xu, Chen, Srinivasan, Freitas, Doucet,
  and Gretton]{Xu.2021}
Liyuan Xu, Yutian Chen, Siddarth Srinivasan, Nando~de Freitas, Arnaud Doucet,
  and Arthur Gretton.
\newblock Learning deep features in instrumental variable regression.
\newblock In \emph{ICLR}, 2021{\natexlab{a}}.

\bibitem[Xu et~al.(2021{\natexlab{b}})Xu, Kanagawa, and Gretton]{Xu.2021b}
Liyuan Xu, Heishiro Kanagawa, and Arthur Gretton.
\newblock Deep proxy causal learning and its application to confounded bandid
  policy evaluation.
\newblock In \emph{NeurIPS}, 2021{\natexlab{b}}.

\bibitem[Yang \& Tokdar(2015)Yang and Tokdar]{Yang.2015}
Yun Yang and Surya~T. Tokdar.
\newblock Minimax-optimal nonparametric regression in high dimensions.
\newblock \emph{The Annals of Statistics}, 43\penalty0 (2):\penalty0 652--674,
  2015.

\bibitem[Yazdani \& Boerwinkle(2015)Yazdani and Boerwinkle]{Yazdani.2015}
Azam~M. Yazdani and Eric Boerwinkle.
\newblock Causal inference in the age of decision medicine.
\newblock \emph{Journal of Data Mining in Genomics {\&} Proteomics}, 6\penalty0
  (1), 2015.

\bibitem[Yoon et~al.(2018)Yoon, Jordon, and {van der Schaar}]{Yoon.2018}
Jinsung Yoon, James Jordon, and Mihaela {van der Schaar}.
\newblock Ganite: Estimation of individualized treatment effects using
  generative adversarial nets.
\newblock In \emph{ICLR}, 2018.

\bibitem[Zhang et~al.(2020)Zhang, Bellot, and {van der Schaar}]{Zhang.2020}
Yao Zhang, Alexis Bellot, and Mihaela {van der Schaar}.
\newblock Learning overlapping representations for the estimation of
  individualized treatment effects.
\newblock In \emph{AISTATS}, 2020.

\end{thebibliography}
\bibliographystyle{iclr2023_conference}

\clearpage

\appendix

\section{Extended related work}
\label{app:rel_work}
\textbf{CATE methods without unconfoundedness:} Various machine learning methods for estimating CATEs \emph{without} unobserved confounding have been proposed in recent literature \citep{Alaa.2017,Curth.2021,Kunzel.2019,Lim.2018,Shalit.2017,Wager.2018,Yoon.2018, Zhang.2020}. To remove plug-in bias, the DR-learner performs a second stage regression on the uncentered influence function of the average treatment effect \citep{Kennedy.2022d}. However, under unobserved confounding, all of these methods are biased (see Appendix~\ref{app:baseline}). As a result, this hampers their performance in our setting. 

\textbf{Non-IV methods for unobserved confounding:} 
There is a rich literature for causal effect estimation under unobserved confounding, which does not assume the existence of instrumental variables. Methods include deconfounding methods \citep{Wang.2019, Bica.2020b}, proxy learning methods \citep{Cui.2020, Xu.2021b}, and causal sensitivity analysis \citep{Kallus.2019, Jesson.2021}..

\textbf{Classical IV methods:}
IV methods address the problem of unobserved confounding by exploiting the variance in treatment and outcome induced by the instruments. Traditionally, two-stage least squares (2SLS) has been used for estimating causal effects \citep{Wright.1928, Angrist.1991}. 2SLS was originally developed in economics, and follows a two-stage procedure: it performs a first stage regression of treatment $A$ on the instrument $Z$, and then uses the fitted values for a second stage regression to predict the outcome $Y$. Several nonparametric methods have been developed in econometric to generalize 2SLS in order to account for non-linearities within the data \citep{Newey.2003, Wang.2021}, yet these are limited to low-dimensional settings.

\clearpage

\section{Proofs}
\label{app:proofs}

We start by deriving an auxiliary Lemma. That is, we derive an explicit expression for the Stage~2 oracle pseudo outcome regression $\E[\hat{Y}_{\mathrm{MR}} \mid X=x]$ of \frameworkname.
\begin{lemma}\label{lem:second_stage}
\begin{equation}
\begin{split}
    & \E[\hat{Y}_{\mathrm{MR}} \mid X=x] \\
    = & 
   \frac{\pi(x)}{\hat{\delta}_A(x) \hat{\pi}(x)} \left(\mu_1^Y(x) - \mu_1^A(x) \, \hat{\tau}_{\mathrm{init}}(x)\right)
   + \frac{(1-\pi(x))}{\hat{\delta}_A(x)(1 - \hat{\pi}(x))} \left(\mu_0^A(x) \, \hat{\tau}_{\mathrm{init}}(x) -\mu_0^Y(x) \right)\\ & \quad + 
   \frac{\hat{\mu}_0^A(x) \, \hat{\tau}_{\mathrm{init}}(x) - \hat{\mu}_0^Y (x)}{\hat{\delta}_A(x)} \left(\frac{\pi(x)}{\hat{\pi}(x)} - \frac{1- \pi(x)}{1-\hat{\pi}(x)} \right) + \hat{\tau}_{\mathrm{init}}(x)   
\end{split}
\end{equation}
\end{lemma}

\begin{proof}
\begin{align}
 &\E[\hat{Y}_{\mathrm{MR}} \mid X=x]\\ =& \pi(x) \E\left[ \frac{Y - A \, \hat{\tau}_{\mathrm{init}}(X) - \hat{\mu}_0^Y(X) + \hat{\mu}_0^A(X) \, \hat{\tau}_{\mathrm{init}}(X)} {\hat{\delta}_A(X) \, \hat{\pi}(X)} \; \middle| \; X = x, Z = 1\right] \notag\\ 
 & \quad + (1 - \pi(x)) \E\left[ \frac{Y - A \, \hat{\tau}_{\mathrm{init}}(X) - \hat{\mu}_0^Y(X) + \hat{\mu}_0^A(X) \, \hat{\tau}_{\mathrm{init}}(X)}{\hat{\delta}_A(X) \, (1 - \hat{\pi}(X))} \; \middle| \; X = x, Z = 0\right] + \hat{\tau}_{\mathrm{init}}(x)\\
  =& \frac{\pi(x)}{\hat{\delta}_A(x) \, \hat{\pi}(x)} \left( \mu_1^Y(x) - \mu_1^A(x) \, \hat{\tau}_{\mathrm{init}}(x) - \hat{\mu}_0^Y(x) + \hat{\mu}_0^A(x) \, \hat{\tau}_{\mathrm{init}}(x) \right) \notag \\
 & \quad + \frac{1-\pi(x)}{\hat{\delta}_A(x) \, (1-\hat{\pi}(x))} \left(\mu_0^Y(x) - \mu_0^A(x) \, \hat{\tau}_{\mathrm{init}}(x) - \hat{\mu}_0^Y(x) + \hat{\mu}_0^A(x) \, \hat{\tau}_{\mathrm{init}}(x) \right) + \hat{\tau}_{\mathrm{init}}(x)
\end{align}

Rearranging the terms yields the desired result.
\end{proof}

\subsection{Proof of Theorem~\ref*{thrm:robustness} (multiple robustness property)}

We use Lemma~\ref{lem:second_stage} to show that under each of the three conditions it follows that $\E[\hat{Y}_{\mathrm{MR}} \mid X=x] = \tau(x)$.
\begin{enumerate}
\item \begin{align}
    &\E[\hat{Y}_{\mathrm{\mathrm{MR}}} \mid X=x] \\ =&  \frac{\pi(x)}{\hat{\delta}_A(x) \, \hat{\pi}(x)} \left(\mu_1^Y(x) - \mu_1^A(x) \, \tau(x) + \mu_0^A(x) \, \tau(x) - \mu_0^Y (x)\right) \notag \\
   & \quad + \frac{(1-\pi(x))}{\hat{\delta}_A(x) \, (1 - \hat{\pi}(x))} \left(\mu_0^A(x) \, \tau(x) -\mu_0^Y(x)  - \mu_0^A(x) \, \tau(x) + \mu_0^Y (x)\right)
    + \tau(x) \\
    =& \frac{\pi(x)}{\hat{\delta}_A(x) \, \hat{\pi}(x)} \left(\delta_Y(x) - \delta_Y(x) \right) + \tau(x) = \tau(x).
    \end{align}
\item \begin{align}\E[\hat{Y}_{\mathrm{\mathrm{MR}}} \mid X=x] &= \frac{\left(\mu_1^Y(x) - \mu_1^A(x) \, \hat{\tau}_{\mathrm{init}}(x)\right)}{\delta_A(x)}
   + \frac{\left(\mu_0^A(x) \, \hat{\tau}_{\mathrm{init}}(x) -\mu_0^Y(x) \right)}{\delta_A(x)} + \hat{\tau}_{\mathrm{init}}(x)  \\ &= \frac{\delta_Y(x) - \hat{\tau}_{\mathrm{init}}(x) \, \delta_A(x)}{\delta_A(x)} + \hat{\tau}_{\mathrm{init}}(x) = \tau(x).
   \end{align}
  
\item \begin{align}\E[\hat{Y}_{\mathrm{\mathrm{MR}}} \mid X=x] &= \frac{\left(\mu_1^Y(x) - \mu_1^A(x) \, \tau(x)\right)}{\hat{\delta}_A(x)}
   + \frac{\left(\mu_0^A(x) \, \tau(x) -\mu_0^Y(x) \right)}{\hat{\delta}_A(x)} + \tau(x) \\ &= \frac{\delta_Y(x)}{\hat{\delta}_A(x)} - \tau(x) \frac{\delta_A(x)}{\hat{\delta}_A(x)} + \tau(x) = \tau(x)
   \end{align}
\end{enumerate}

\subsection{Proof of Theorem~\ref*{thrm:upperbound} (Convergence rate of \frameworkname)}

To prove Theorem~\ref*{thrm:upperbound}, we need an additional assumption on the second stage regression estimator $\hat{\E}_n$. We refer to \citet{Kennedy.2022d} (Proposition 1) for a detailed discussion on this assumption.

\begin{assumption}[From Proposition~1 of \citet{Kennedy.2022d}]\label{ass:kennedy}
Let $Y_{\mathrm{\mathrm{MR}}}$ be the corresponding oracle to the estimated pseudo-outcome $\hat{Y}_{\mathrm{\mathrm{MR}}}$. We assume that the pseudo-regression estimator $\hat{\E}_n$ satisfies
\begin{equation}
    \frac{\hat{\E}_n[\hat{Y}_{\mathrm{MR}} \mid X = x] - \hat{\E}_n[{Y}_{\mathrm{MR}} \mid X = x] - \hat{\E}_n[\hat{Y}_{\mathrm{\mathrm{MR}}} - {Y}_{\mathrm{\mathrm{MR}}} \mid X = x]}{\sqrt{\E\left[\left(\hat{\E}_n[{Y}_{\mathrm{MR}} \mid X = x] - \E[{Y}_{\mathrm{MR}} \mid X = x] \right)^2\right]}} \overset{p}{\to} 0
\end{equation}
and 
\begin{equation}
  \E\left[\hat{\E}_n[\hat{r}(X) \mid X = x]^2 \right]  \asymp \E\left[\hat{r}(x)^2 \right],
\end{equation}
where $r(x) = \E[\hat{Y}_{\mathrm{\mathrm{MR}}} \mid X = x] - \tau(x)$
\end{assumption}

To prove Theorem~\ref{thrm:upperbound}, we derive a more general bound the depends on the pointwise mean squared errors of the nuisance estimators. Theorem~\ref{thrm:upperbound} follows immediately by applying Assumption~\ref*{ass:smoothness}.
\begin{lemma}\label{lem:general_bound}
Consider the setting described in Theorem~\ref{thrm:upperbound}. Then, 
\begin{align}\label{eq:upperbound_general}
   & \E\left[\left(\hat{\tau}_{\mathrm{init}}(x) - \tau(x)\right)^2\right] \\  \lesssim & \, \mathcal{R}(x)  + \E\left[\left(\hat{\tau}_{\mathrm{init}}(x) - \tau(x)\right)^2\right] \left( \E\left[\left(\hat{\delta}_A(x) - \delta_A(x)\right)^2\right] + \E\left[\left(\hat{\pi}(x) - \pi(x)\right)^2\right]\right) \notag
    \\ & \quad + \E\left[\left(\hat{\pi}(x) - \pi(x)\right)^2\right] \left(\E\left[\left(\hat{\mu}_0^Y(x) - \mu_0^Y(x)\right)^2\right] 
     + \E\left[\left(\hat{\mu}_0^A(x) - \mu_0^A(x)\right)^2\right] \right).
\end{align}
\end{lemma}

\begin{proof}
Let $Y_{\mathrm{\mathrm{MR}}}$ be the corresponding oracle to $\hat{Y}_{\mathrm{\mathrm{MR}}}$ and define $\widetilde{\tau}_{\mathrm{MRIV}}(x) = \hat{\E}_n[Y_{\mathrm{\mathrm{MR}}} \mid X = x]$.
Using Assumption~\ref{ass:kennedy}, we can apply Proposition~1 of \citet{Kennedy.2022d} and obtain
\begin{equation}
    \E\left[\left(\hat{\tau}_{\mathrm{init}}(x) - \tau(x)\right)^2\right] \lesssim  \mathcal{R}(x) + \E\left[\hat{r}(x)^2 \right],
\end{equation}
where $\mathcal{R}(x) = \E\left[\left(\widetilde{\tau}_{\mathrm{\mathrm{MR}}}(x) - \tau(x)\right)^2\right]$ is the oracle risk of the second stage regression.
We can apply Lemma~\ref{lem:second_stage} to obtain
\begin{align}
   \hat{r}(x) &=  \frac{\pi(x)}{\hat{\delta}_A(x) \, \hat{\pi}(x)} \left(\mu_1^Y(x) - \mu_1^A(x) \, \hat{\tau}_{\mathrm{init}}(x)\right)
   + \frac{(1-\pi(x))}{\hat{\delta}_A(x) \, (1 - \hat{\pi}(x))} \left(\mu_0^A(x) \, \hat{\tau}_{\mathrm{init}}(x) -\mu_0^Y(x) \right) \notag \\ & \quad + 
   \frac{\hat{\mu}_0^A(x) \, \hat{\tau}_{\mathrm{init}}(x) - \hat{\mu}_0^Y (x)}{\hat{\delta}_A(x)} \left(\frac{\pi(x)}{\hat{\pi}(x)} - \frac{1- \pi(x)}{1-\hat{\pi}(x)} \right) + \hat{\tau}_{\mathrm{init}}(x) - \tau(x) \\
   &= \left(\frac{\mu_1^Y(x) - \mu_0^Y(x)}{\hat{\delta}_A(x)} \right) \frac{\pi(x)}{\hat{\pi}(x)} + \frac{\mu_0^Y(x) - \hat{\mu}_0^Y(x)}{\hat{\delta}_A(x)} \left(  \frac{\pi(x)}{\hat{\pi}(x)} -  \frac{1 - \pi(x)}{1 - \hat{\pi}(x)} \right)  + \left(\hat{\tau}_{\mathrm{init}}(x) - \tau(x)\right) \notag \\
   & \quad + \left( \frac{(\mu_0^A(x) - \mu_1^A(x)) \, \hat{\tau}_{\mathrm{init}}(x)}{\hat{\delta}_A(x)} \right) \frac{\pi(x)}{\hat{\pi}(x)} + \frac{(\hat{\mu}_0^D(x) - \mu_0^D(x)) \, \hat{\tau}_{\mathrm{init}}(x)}{\hat{\delta}_A(x)} \left(  \frac{\pi(x)}{\hat{\pi}(x)} -  \frac{1 - \pi(x)}{1 - \hat{\pi}(x)} \right) \\
   &= \frac{\delta_Y(x) \, \pi(x)}{\hat{\delta}_A(x) \, \hat{\pi}(x)} + \frac{\left(\mu_0^Y(x) - \hat{\mu}_0^Y(x)\right) \left(\pi(x) - \hat{\pi}(x)\right)}{\hat{\delta}_A(x) \, \hat{\pi}(x) \left(1 - \hat{\pi}(x) \right)} + \left(\hat{\tau}_{\mathrm{init}}(x) - \tau(x)\right) \notag \\
   & \quad - \frac{\delta_A(x) \, \pi(x) \, \hat{\tau}_{\mathrm{init}}(x)}{\hat{\delta}_A(x) \, \hat{\pi}(x)} + \frac{\left(\hat{\mu}_0^A(x) - \mu_0^A(x)\right) \hat{\tau}_{\mathrm{init}}(x) \left(\pi(x) - \hat{\pi}(x) \right)}{\hat{\delta}_A(x) \, \hat{\pi}(x) \left(1 - \hat{\pi}(x)\right)} \\
   &=\frac{\left(\pi(x) - \hat{\pi}(x) \right)}{\hat{\delta}_A(x) \, \hat{\pi}(x) \left(1 - \hat{\pi}(x)\right)}\left[\left(\mu_0^Y(x) - \hat{\mu}_0^Y(x)\right) + \left(\hat{\mu}_0^A(x) - \mu_0^A(x)\right)\hat{\tau}_{\mathrm{init}}(x)\right] \notag \\
   & \quad + \left(\hat{\tau}_{\mathrm{init}}(x) - \tau(x)\right) + \frac{\pi(x) \delta_A(x)}{\hat{\pi}(x) \hat{\delta}_A(x)} \left(\tau(x) - \hat{\tau}_{\mathrm{init}}(x)\right) \\
   &= \frac{\left(\pi(x) - \hat{\pi}(x) \right)}{\hat{\delta}_A(x) \, \hat{\pi}(x) \left(1 - \hat{\pi}(x)\right)}\left[\left(\mu_0^Y(x) - \hat{\mu}_0^Y(x)\right) + \left(\hat{\mu}_0^A(x) - \mu_0^A(x)\right)\hat{\tau}_{\mathrm{init}}(x)\right] \notag \\
   & \quad + \left(\tau(x) - \hat{\tau}_{\mathrm{init}}(x)\right) \left(\delta_A(x) - \hat{\delta}_A(x)\right) \pi(x) + \left(\tau(x) - \hat{\tau}_{\mathrm{init}}(x)\right) \left(\pi(x) - \hat{\pi}(x)\right) \hat{\delta}_A(x).
\end{align}

Applying the inequality $(a+b)^2 \leq 2(a^2 + b^2)$ together with Assumption~\ref*{ass:boundedness} and the fact that $\pi(x) \leq 1$ yields
\begin{align}
    \hat{r}(x)^2 & \leq \frac{4}{\epsilon^4 \rho^2} \left(\pi(x) - \hat{\pi}(x) \right)^2 \left[\left(\mu_0^Y(x) - \hat{\mu}_0^Y(x)\right)^2 + \left(\hat{\mu}_0^A(x) - \mu_0^A(x)\right)^2 K^2\right] \notag\\
    & \quad +4 \left(\tau(x) - \hat{\tau}_{\mathrm{init}}(x)\right)^2 \left(\delta_A(x) - \hat{\delta}_A(x)\right)^2 + 4\left(\tau(x) - \hat{\tau}_{\mathrm{init}}(x)\right)^2 \left(\pi(x) - \hat{\pi}(x)\right)^2 .
\end{align}

By setting $\widetilde{K} = \max\{K,1\}$, we obtain
\begin{align}
    \hat{r}(x)^2 & \leq \frac{4\widetilde{K}^2}{\epsilon^4 \rho^2}\left( \left(\pi(x) - \hat{\pi}(x) \right)^2 \left[\left(\mu_0^Y(x) - \hat{\mu}_0^Y(x)\right)^2 + \left(\hat{\mu}_0^A(x) - \mu_0^A(x)\right)^2 + \left(\hat{\tau}_{\mathrm{init}}(x) - \tau(x)\right)^2\right] \right. \notag\\
    & \quad + \left. \left(\tau(x) - \hat{\tau}_{\mathrm{init}}(x)\right)^2 \left(\delta_A(x) - \hat{\delta}_A(x)\right)^2 \right).
\end{align}

Applying expectations on both sides yields the results, because $(\hat{\pi}(x), \hat{\delta}_A(x)) \indep (\hat{\mu}_0^Y(x), \hat{\mu}_0^A(x), \hat{\tau}_{\mathrm{init}}(x))$ due to sample splitting.
\end{proof}

\subsection{Proof of Theorem~\ref*{thrm:rate_wald} (Convergence rate of the Wald estimator)}

\begin{proof}
We define $\widetilde{C} = \max\{C,1\}$ and obtain the upper bound
\begin{align}\label{eq:rate_wald_general}
      &  (\hat{\tau}_W(x) - \tau(x))^2\\ =& \left(\frac{(\hat{\mu}_1^Y(x) - \mu_1^Y(x)) \, \delta_A(x) + (\mu_0^Y(x) - \hat{\mu}_0^Y(x)) \, \delta_A(x) + (\delta_A(x) - \hat{\delta}_A(x)) \, \delta_Y(x)}{\delta_A(x) \, \hat{\delta}_A(x)}\right)^2 \\
         \leq & \frac{4 \widetilde{C}^2}{\rho^2 \widetilde{\rho}^2} \left[(\hat{\mu}_1^Y(x) - \mu_1^Y(x))^2 + (\hat{\mu}_0^Y(x) - \mu_0^Y(x))^2 + (\delta_A(x) - \hat{\delta}_A(x))^2 \right]\\
        \leq & \frac{8 \widetilde{C}^2}{\rho^2 \widetilde{\rho}^2} \left[(\hat{\mu}_1^Y(x) - \mu_1^Y(x))^2 + (\hat{\mu}_0^Y(x) - \mu_0^Y(x))^2 + (\hat{\mu}_1^A(x) - \mu_1^A(x))^2 \right. \notag \\ & \quad \quad  \quad \left.+ (\hat{\mu}_0^A(x) - \mu_0^A(x))^2 \right],
\end{align}
where we used the inequality $(a+b)^2 \leq 2(a^2 + b^2)$ several times.
Taking expectations and applying the smoothness assumptions yields the result.
\end{proof}

\clearpage

\section{Theoretical analysis under sparsity assumptions}\label{app:sparsity}

In Sec.~\ref{sec:theory}, we analyzed \frameworkname theoretically by imposing smoothness assumptions on the underlying data generating process. In particular, we derived a multiple robust convergence rate and showed that \frameworkname outperforms the Wald estimator if the oracle CATE is smoother than its components. In this section, we derive similar results by relying on a different set of assumptions. Instead of using smoothness, we make assumptions on the level of sparsity of the CATE components. This assumption is often imposed in high-dimensional settings ($n < p $) and is in line with previous literature on analyzing CATE estimators \citep{Curth.2021, Kennedy.2022d}.

In the following, we say a function $f(x)$ is $k$-sparse, if it is linear in $x \in \R^p$ and it only depends on $k < \min\{n, p\}$ predictors. \citep{Yang.2015} showed, that in this case the minimax rate of $f(x)$ is given by $\frac{k \log(p)}{n}$. The linearity assumption can be relaxed to an additive structural assumption, which we omit here for simplicity. In the following, we replace the smoothness conditions in Assumption~\ref{ass:smoothness} with sparsity conditions.

\begin{assumption}[Sparsity]\label{ass:sparsity}
We assume that (1)~the nuisance components $\mu_i^Y(\cdot)$ are $\alpha$-sparse, $\mu_i^A(\cdot)$ and $\delta_A(\cdot)$ are $\beta$-sparse, and $\pi(\cdot)$ is $\delta$-sparse; (2)~all nuisance components are estimated with their respective minimax rate of $\frac{k \log(p)}{n}$, where $k \in \{\alpha, \beta, \delta\}$; and (3)~the oracle CATE $\tau(\cdot)$ is $\gamma$-sparse and the initial CATE estimator $\hat{\tau}_{\mathrm{init}}$ converges with rate $r_{\tau}(n)$.
\end{assumption}

We restate now our result from Theorem~\ref{thrm:rate_wald} for \frameworkname using the sparsity assumption.

\begin{theorem}[\frameworkname upper bound under sparsity]\label{thrm:upperbound_sparsity}
We consider the same setting as in Theorem~\ref*{thrm:upperbound} under the sparsity assumption \ref{ass:sparsity}.
If the second-stage estimator $\hat{\E}_n$ yields the minimax rate $\frac{\gamma \log(p)}{n}$ and satisfies Assumption~\ref{ass:kennedy}, the oracle risk is upper bounded by
\begin{equation*}\label{eq:upperbound_mriv_sparsity}
\begin{split}
    \E\left[\left(\hat{\tau}_{\mathrm{MRIV}}(x) - \tau(x)\right)^2\right] & \lesssim  \frac{\gamma \log(p)}{n} + r_{\tau}(n) \frac{(\beta + \delta) \log(p)}{n} + 
    \frac{(\alpha + \beta) \delta \log^2(p)}{n^2}.
\end{split}
\end{equation*}
\end{theorem}
\begin{proof}
Follows immediately from Lemma~\ref{lem:general_bound} by applying Ass-~\ref{ass:sparsity}.
\end{proof}

Again, we obtain a multiple robust convergence rate for \frameworkname in the sense that \frameworkname achieves a fast rate even if the initial estimator or several nuisance estimators converge slowly. More precisely, for a fast convergence rate of $\hat{\tau}_{\mathrm{MRIV}}(x)$, it is sufficient if either: (1)~$r_{\tau}(n)$ decreases fast and $\delta$ is small; (2)~$r_{\tau}(n)$ decreases fast and $\alpha$ and $\beta$ are small; or (3)~all $\alpha$, $\beta$, and $\delta$ are small. 

We derive now the corresponding rate for the Wald estimator.

\begin{theorem}[Wald oracle upper bound]\label{thrm:rate_wald_sparsity}
Given estimators $\hat{\mu}_i^Y(x)$ and $\hat{\mu}_i^A(x)$. Let $\hat{\delta}_A(x) = \hat{\mu}_1^A(x) - \hat{\mu}_0^A(x)$ satisfy Assumption~\ref{ass:boundedness}. Then, under Assumption~\ref{ass:sparsity} the oracle risk of the Wald estimator $\hat{\tau}_W(x)$ is bounded by
\begin{equation}\label{eq:upperbound_wald_sparsity}
    \E\left[(\hat{\tau}_{\mathrm{W}}(x) - \tau(x))^2\right] \lesssim  \frac{(\alpha + \beta) \log(p)}{n}
\end{equation}
\end{theorem}
\begin{proof}
Follows immediately from the proof of Theorem~\ref{thrm:rate_wald}, i.e., from Eq.\eqref{eq:rate_wald_general} by applying Ass-~\ref{ass:sparsity}.
\end{proof}

If $\alpha = \beta = \delta$, we obtain the rates
\begin{equation}
 \E\left[\left(\hat{\tau}_{\mathrm{MRIV}}(x) - \tau(x)\right)^2\right] \lesssim  \frac{\gamma \log(p)}{n} + \frac{\alpha^2 \log^2(p)}{n^2} \quad \text{and} \quad
        \E\left[(\hat{\tau}_{\mathrm{W}}(x) - \tau(x))^2\right] \lesssim  \frac{\alpha \log(p)}{n},
\end{equation}

which means that $\hat{\tau}_{\mathrm{MRIV}}(x)$ outperforms $\hat{\tau}_{\mathrm{W}}(x)$ for $\gamma < \alpha$, i.e., if the oracle CATE is more sparse than its components.

\clearpage

\section{Mathematical details regarding Assumption~\ref{ass:smoothness}}\label{app:rates}
In this section, we briefly state the formal definitions of the convergence ratesin Assumption~\ref{ass:smoothness}. We follow \citet{Stone.1980}. Let $\theta$ be a parameter and $T(\theta)$ some target functional that we want to estimate.
\begin{definition}
A sequence of estimators $\hat{T}_n(\theta)$ of a functional $T(\theta)$ converges with rate $r_{T}(n)$ if
\begin{equation}
   \lim_{c \to 0} \liminf_{n \in \N} \sup_{\theta \in \Theta} \mathbb{P}_\theta( \hat{T}_n(\theta) - T(\theta)| > c r_{\theta}(n)) = 0. 
\end{equation}
\end{definition}
\begin{definition}
A rate $r_{T}(n)$ is called an upper bound to the rate of convergence if for all estimators $\hat{T}_n(\theta)$ it holds for all $c > 0$ that
\begin{equation}
    \liminf_{n \in \N} \sup_{\theta \in \Theta} \mathbb{P}_\theta( |\hat{T}_n(\theta) - T(\theta)|| > c r_{\theta}(n)) > 0. 
\end{equation}
and \begin{equation}
   \lim_{c \to 0} \liminf_{n \in \N} \sup_{\theta \in \Theta} \mathbb{P}_\theta( \hat{T}_n(\theta) - T(\theta)| > c r_{\theta}(n)) = 1. 
\end{equation}
$r_{T}(n)$ is called optimal if it is both achievable and an upper bound.
\end{definition}
\citet{Stone.1980} showed that for a nonparametric regression problem with an $\eta$-smooth regression function, the optimal rate of convergence is $n^{-\frac{2\eta}{2\eta + p}}$, where $p$ is the dimension of the covariate space.

\clearpage

\section{Simulated data}\label{app:sim}

In the following, we describe how we simulate synthetic data for the experiments in Sec.~\ref{sec:exp_sim} from the main paper. As mentioned therein, we simulate the CATE components from Gaussian processes using the prior induced by the Matern kernel \citep{Rasmussen.2008}
\begin{equation}
    K_{\ell, \nu}(x_i, x_j) = \frac{1}{\Gamma(\nu)2^{\nu-1}}\left( \frac{\sqrt{2\nu}}{\ell} \| x_i - x_j\|_2  \right)^\nu K_\nu \left(\frac{\sqrt{2\nu}}{\ell} \| x_i - x_j\|_2 \right),
\end{equation}
where $\Gamma(\cdot)$ is the Gamma function and $K_\nu(\cdot)$ is the modified Bessel function of second kind. Here, $\ell$ is the length scale of the kernel and $\nu$ controls the smoothness of the sampled functions. 

We set $\ell = 1$ and sample functions $\delta_Y \sim \mathcal{GP}(0, K_{\ell, \gamma})$, $\mu_0^Y \sim \mathcal{GP}(0, K_{\ell, \alpha})$, $f_1 \sim \mathcal{GP}(0, K_{\ell, \beta})$, $f_0 \sim \mathcal{GP}(0, K_{\ell, \beta})$ and $g \sim \mathcal{GP}(0, K_{\ell, \beta})$. 
Then, we define $\mu_1^Y = \delta_Y + \mu_0^Y$, $\mu_1^A = 0.3 \cdot \sigma \circ f_1 + 0.7$, $\mu_0^A = 0.3 \cdot \sigma \circ f_0$, $\delta_A = \mu_1^A - \mu_0^A$, $\mu_0^Y = c_0 \delta_A$, and $\pi = \sigma \circ g$. Finally, we set the oracle CATE to 
\begin{equation}
 \tau = \frac{\mu_1^Y - \mu_0^Y}{\mu_1^A - \mu_0^A} = \frac{\delta_Y}{\delta_A}.
\end{equation}
Note that we can create a setup where the CATE $\tau$ is smoother than its components by using a small $\alpha / \beta$ ratio. An example is shown in Fig.~\ref{fig:sim}.

\begin{figure}[h]
\centering
\includegraphics[width=0.6\linewidth]{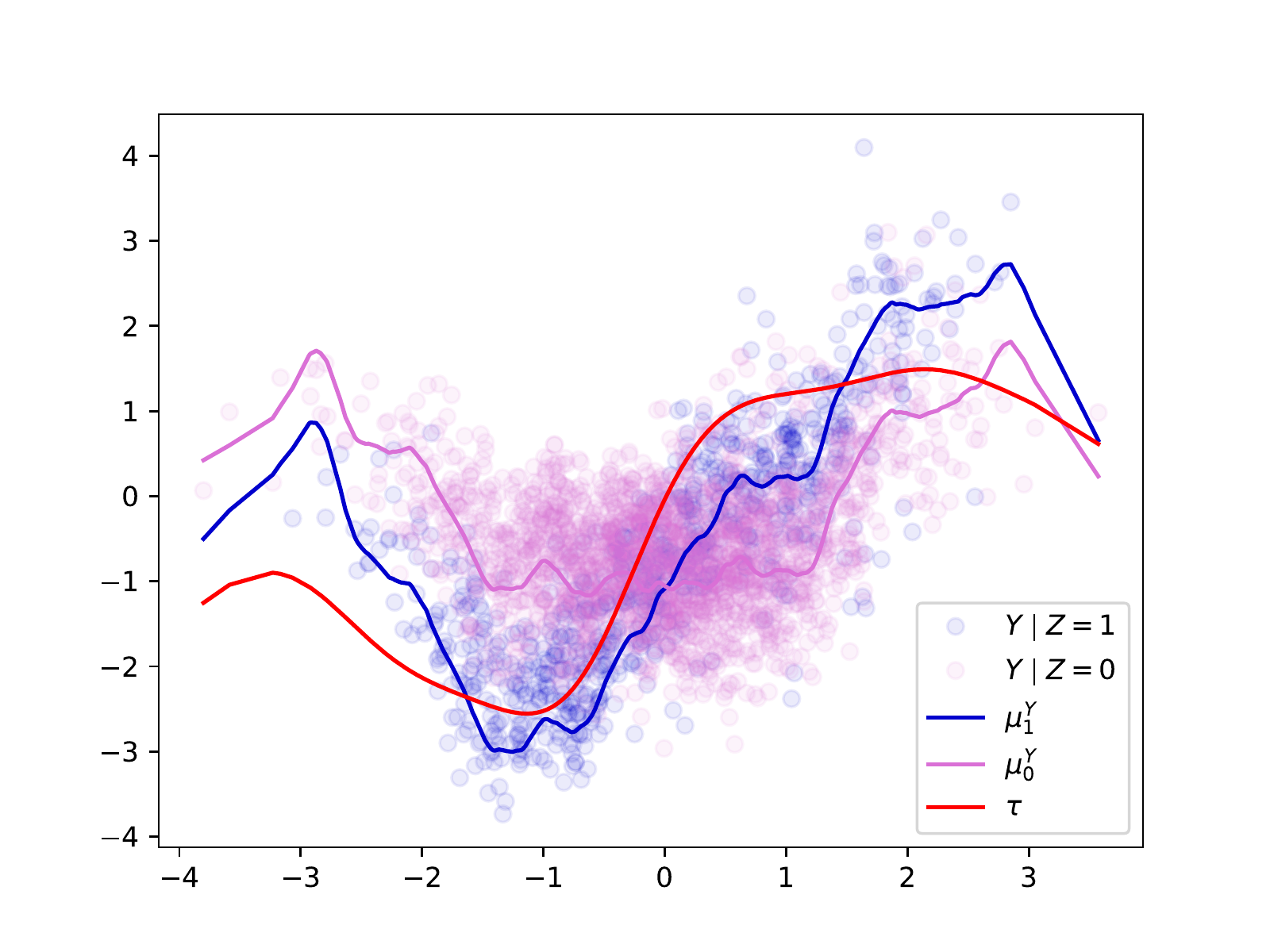}
\vspace{-0.3cm}
\caption{Gaussian process simulation for $\alpha = 1.5$ and $\beta = 50$.}
\label{fig:sim}
\end{figure}

In the following, we describe how we generate data the $(X, Z, A, Y)$ using the CATE components $\mu_i^Y(x)$, $\mu_i^A(x)$, and $\pi(x)$. We begin by sampling $n$ observed confounder $X \sim \mathcal{N}(0,1)$, unobserved confounders $U \sim \mathcal{N}\left(0, 0.2^2\right)$, and instruments $Z \sim \mathrm{Bernoulli}(\pi(X))$. Then, we obtain treatments via
\begin{equation}\label{eq:sim_treat}
     A = Z \, \mathbbm{1}\{U + \epsilon_{A} > \alpha_1(X)\} + (1-Z) \, \mathbbm{1}\{U + \epsilon_{A} > \alpha_0(X)\}
\end{equation}
with indicator function $\mathbbm{1}$, noise $\epsilon_{A} \sim \mathcal{N}\left(0, 0.1^2\right)$, and $\alpha_i(X) = \Phi^{-1}\left(1 - \mu_i^A(X)\right) \sqrt{0.1^2 + 0.2^2}$, where $\Phi^{-1}$ denotes the quantile function of the standard normal distribution. Finally, we generate the outcomes via

\begin{align}\label{eq:sim_outcome}
     Y &= A \left( \frac{(\mu_1^A(X) - 1)\mu_0^Y(X) - \mu_0^A(X)\mu_1^Y(X) + \mu_1^Y(X)}{\delta_A(X)}\right) \\& \quad + (1-A) \left( \frac{\mu_1^A(X)\mu_0^Y(X) - \mu_0^A(X)\mu_1^Y(X)}{\delta_A(X)} \right) + \alpha_U U +\epsilon_Y,
\end{align}
where $\epsilon_Y \sim \mathcal{N}\left(0, 0.3^2\right)$ is noise and $\alpha_U> 0$ is a parameter indicating the level of unobserved confounding. 
This choice of $A$ and $Y$ in Eq.~\eqref{eq:sim_treat} and Eq.~\eqref{eq:sim_outcome}, respectively, implies that $\tau(x)$ is indeed the CATE, \ie, it holds that $\tau(x) = \E[Y(1) -Y(0) \mid X = x]$. 

\begin{lemma}\label{lem:sim}
Let $(X, Z, A, Y)$ be sampled from the the previously described procedure. Then, it holds that 
\begin{equation}
    \mu_i^A(x) = \E[A \mid Z=i, X = x] \quad \text{and} \quad \mu_i^Y(x) = \E[Y \mid Z=i, X = x].
\end{equation}
\end{lemma}
\begin{proof}
The first claim follows from
\begin{equation}
    \E[A \mid Z = i, X = x] = \mathbb{P}\left(U + \epsilon_{A} > \alpha_i(x)\right) 
    = 1 - \Phi(\Phi^{-1}(1 - \mu_i^A(x))) = \mu_i^A(x),
\end{equation}
because $U + \epsilon_A \sim \mathcal{N}(0, \sqrt{0.1^2 + 0.2^2})$.
The second claim follows from
\begin{align}
\E[Y \mid Z = i, X = x] &= \mu_i^A(x) \left( \frac{(\mu_1^A(x) - 1)\mu_0^Y(x) - \mu_0^A(x)\mu_1^Y(x) + \mu_1^Y(x)}{\delta_A(x)}\right) \\& \quad + (1-\mu_i^A(x)) \left( \frac{\mu_1^A(x)\mu_0^Y(x) - \mu_0^A(x)\mu_1^Y(x)}{\delta_A(x)} \right) \\ &= \frac{\mu_i^Y(x) \delta_A(x)}{\delta_A(x)} = \mu_i^Y(x).
\end{align}

\end{proof}

\clearpage

\section{Oregon health insurance experiment}\label{app:real}

The so-called \emph{Oregon health insurance experiment}\footnote{Data available here: https://www.nber.org/programs-projects/projects-and-centers/oregon-health-insurance-experiment} (OHIE) \citep{Finkelstein.2012} was an important RCT with non-compliance. It was intentionally conducted as large-scale effort among public health to assess the effect of health insurance on several outcomes such as health or economic status. In 2008, a lottery draw offered low-income, uninsured adults in Oregon participation in a Medicaid program, providing health insurance. Individuals whose names were drawn could decide to sign up for the program.

In our analysis, the lottery assignment is the instrument $Z$, the decision to sign up for the Medicaid program is the treatment $A$, and an overall health score is the outcome $Y$. The outcome was obtained after a period of 12 months during in-person interviews. We use the following covariates $X$: age, gender, language, the number of emergency visits before the experiment, and the number of people the individual signed up with. The latter is used to control for peer effects, and it is important to include this variable in our analysis as it is the only variable influencing the propensity score (see below). We extract $\sim$ 10,000 observations from the OHIE data and plot the histograms of all variables in Fig.~\ref{fig:hist}. We can clearly observe the presence of non-compliance within the data, because the ratio of treated / untreated individuals is much lower than the corresponding ratio for the treatment assignment.

\begin{figure}[h]
\centering
\includegraphics[width=0.8\linewidth]{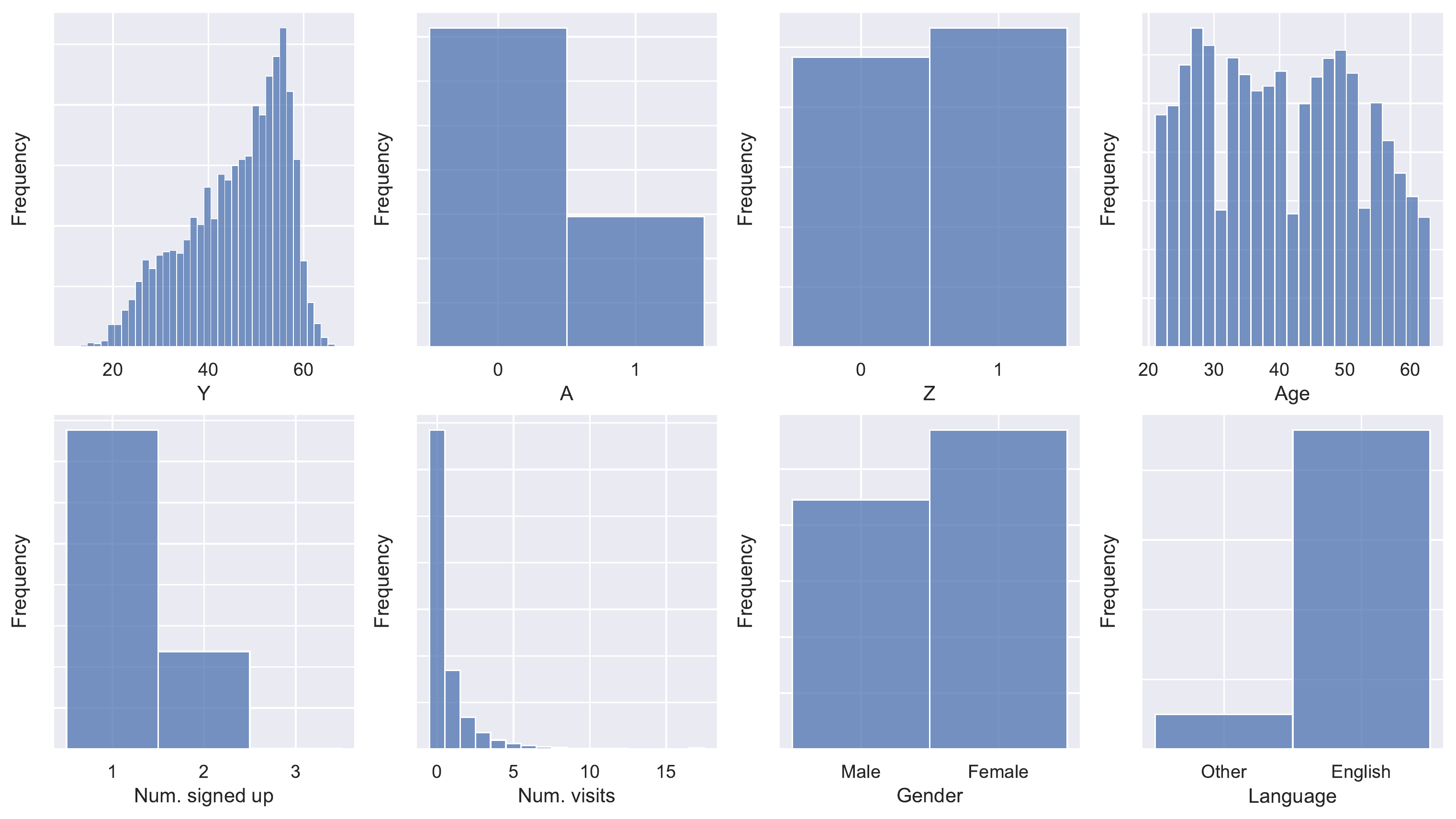}
\caption{Histograms of each variable in our sample from OHIE.}
\label{fig:hist}
\end{figure}

The data collection in the OHIE was done follows: After excluding individuals below the age of 19, above the age of 64, and individuals with residence outside of Oregon, 74,922 individuals were considered for the lottery. Among those, 29,834 were selected randomly and were offered participation in the program. However, the probability of selection depended on the number of household members on the waiting list: for instance, an individual who signed up with another person was twice as likely to be selected. From the 74,922 individuals, 57,528 signed up alone, 17,236 signed up with another person, and 158 signed up with two more people on the waiting list. Thus, the probability of being selected conditional on the number of household members on the waiting list follows the multivariate version of Wallenius' noncentral hypergeometric distribution \citep{Chesson.1976}.

\textbf{Propensity score:} We computed the propensity score as follows. To account for the Wallenius' noncentral hypergeometric distribution, we use the R package \emph{BiasedUrn} to calculate the propensity score $\pi(x) = \mathbb{P}(Z = 1 \mid X = x)$. We obtained 
\begin{equation}\label{eq:ohie_propensity}
    \pi(x) = \begin{cases} 
      0.345, & \text{if individual } x \text{ signed up alone,}\\
      0.571, & \text{if individual } x \text{ signed up with one more person,} \\
      0.719, & \text{if individual } x \text{ signed up with two more people.} 
   \end{cases}
\end{equation}
During the training of both \frameworkname and DRIV, we use the calculated values from Eq.~\eqref{eq:ohie_propensity} for the propensity score. 


\clearpage

\section{Details for baseline methods}\label{app:baseline}

In this section, we give a brief overview on the baselines which we used in our experiments. We implemented: (1)~ CATE methods for unconfoundedness: \textbf{TARNet} \citep{Shalit.2017} and TARNet combined with the \textbf{DR-learner} \citep{Kennedy.2022d}; (2)~general IV methods, i.e., IV methods developed for IV settings with multiple or continuous instruments and treatments: \textbf{2SLS} \citep{Wright.1928}, kernel IV (\textbf{KIV}) \citep{Singh.2019}, \textbf{DFIV} \citep{Xu.2021}, \textbf{DeepIV} \citep{Hartford.2017}, \textbf{DeepGMM} \citep{Bennett.2019}, \textbf{DMLIV} \citep{Syrgkanis.2019}, and DMLIV combined with \textbf{DRIV} (as described in \citep{Syrgkanis.2019}); (3)~the (plug-in) Wald estimator using \textbf{linear models} and Bayesian additive regression trees (\textbf{BART}) \citep{Chipman.2010}. Of note, the DR-learner assumes unconfoundedness, which is why we only combine it TARNet in our experiments. In the following, we provide details regarding methods and implementation.

\subsection{CATE methods for unconfoundedness}
Many CATE methods assume \emph{unconfoundedness}, i.e., that all confounders are observed in the data. Formally, the unconfoundedness assumption can be expressed in the potential outcomes framework as
\begin{equation}
    Y(1), Y(0) \indep A \mid X.
\end{equation}
Under unconfoundedness, the CATE is identified as
\begin{equation}\label{eq:unconfoundedness_identification}
    \tau(x) = \mu_1(x) - \mu_0(x) \quad \text{with} \quad \mu_i(x) = \E[Y \mid A = i, X = x].
\end{equation}
Methods that assume unconfoundedness proceed by estimating $\mu_i(x) = \E[Y \mid A = i, X = x]$ from Eq.~\eqref{eq:unconfoundedness_identification}. However, if unobserved confounders $U$ exist, it follows that
\begin{equation}
    \tau(x) = \E[Y \mid A = 1, X = x, U] - \E[Y \mid A = 0, X = x, U] \neq \mu_1(x) - \mu_0(x),
\end{equation}
which means that estimators that assume unconfoundedness are generally biased. Nevertheless, we include two baselines that assume unconfoundedness into our experiments: TARNet \citep{Shalit.2017} and the DR-learner \citep{Kennedy.2022d}.

\textbf{TARNet} \citep{Shalit.2017}:
TARNet \citep{Shalit.2017} is a neural network that estimates the CATE components $\mu_i(x)$ from Eq.~\ref{eq:unconfoundedness_identification} by learning a shared representation $\Phi(x)$ and two potential outcome heads $h_i(\Phi(x))$. We train TARNet by minimizing the loss
\begin{equation}
    \mathcal{L}(\theta) = \sum_{i=1}^n L\left(h_{a_i}(\Phi(x_i, \theta_\Phi), \theta_{h_i}), y_i \right),
\end{equation}
where $\theta = (\theta_{h_1}, \theta_{h_0}, \theta_\Phi)$ denotes the model parameters and $L$ denotes squared loss if $Y$ is continuous or binary cross entropy loss if $Y$ is binary.

\emph{Note regarding balanced representations:} In \citep{Shalit.2017}, the authors propose to add an additional regularization term inspired from domain adaptation literature, which forces TARNet to learn a balanced representation $\Phi(x)$, i.e., that minimizes the distance the treatment and control group in the feature space. They showed that this approach leads to minimization of a generalization bound on the CATE estimation error if the representation is invertible. 

In our experiments, we refrained from learning balanced representations because minimizing the regularized loss from \citep{Shalit.2017} does not necessarily result in an invertible representation and thus may even harm the estimation performance. For a detailed discussion, we refer to \citep{Curth.2021}. Furthermore, by leaving out the regularization, we ensure comparability between the different baselines. If balanced representations are desired, the balanced representation approach could also be extended to \modelname, as we also build \modelname on learning shared representations.

\textbf{DR-learner} \citep{Kennedy.2022d}: The DR-learner \citep{Kennedy.2022d} is a meta learner that takes arbitrary estimators of the CATE componenets $\mu_i$ and the propensity score $\pi(x) = \mathbb{P}(A = 1 \mid X = x)$ as input and performs a pseudo outcome regression by using the pseudo outcome
\begin{equation}
    \hat{Y}_{\mathrm{MR}} = \left(\frac{A}{\hat{\pi}(X)} - \frac{1-A}{1-\hat{\pi}(X)}\right)Y + \left(1 - \frac{A}{\hat{\pi}(X)}\right)\hat{\mu}_1(X) - \left(1 - \frac{1 - A}{1 - \hat{\pi}(X)} \right)\hat{\mu}_0(X) .
\end{equation}
In our experiments, we use TARNet as base method to provide initial estimators $\hat{\mu}_i(X)$. We further learn propensity score estimates $\hat{\pi}(X)$ by adding a seperate representation to TARNet as done in \citep{Shalit.2017}.

\subsection{General IV methods}

\textbf{2SLS} \citep{Wright.1928}: 2SLS \citep{Wright.1928} is a linear two-stage approach. First, the treatments $A$ are regressed on the instruments $Z$ and fitted values $\hat{A}$ are obtained. In the second stage, the outcome $Y$ is regressed on $\hat{A}$. We implement 2SLS using the scikit-learn package. 
 
\textbf{KIV} \citep{Singh.2019}: Kernel~IV \citep{Singh.2019} generalizes 2SLS to nonlinear settings. KIV assumes that the data is generated by
\begin{equation}\label{eq:general_iv_ass}
    Y = f(A) + U,
\end{equation}
where $U$ is an additive unobserved confounder and $f$ is some unknown (potentially nonlinear) structural function. KIV then models the structural function via
\begin{equation}\label{eq:two_stage_features}
    f(a) = \mu^t \psi(a) \quad \text{and} \quad \E[\psi(A) \mid Z = z] = V \phi(z),
\end{equation}
where $\psi$ and$\phi$ are feature maps. Here, kernel ridge regressions instead of linear regressions are used in both stages to estimate $\mu$ and $V$.

Following \citep{Singh.2019} we use the exponential kernel \citep{Rasmussen.2008} and set the length scale to the median inter-point distance. KIV does not provide a direct way to incorporate the observed confounders $X$. Hence, we augment both the instrument and the treatment with $X$, which is consistent with previous work \citep{Bennett.2019, Xu.2021}. We also use two different samples for each stage as recommended in \citep{Singh.2019}.

\textbf{DFIV} \citep{Xu.2021}: DFIV \citep{Xu.2021} is a similar approach KIV in generalizing 2SLS to nonlinear setting by assuming Eq.~\eqref{eq:general_iv_ass} and Eq.~\eqref{eq:two_stage_features}. However, instead of using kernel methods, DFIV models the features maps $\psi_{\theta_A}$ and $\phi_{\theta_Z}$ as neural networks with parameters $\theta_A$ and $\theta_Z$, respectively. DFIV is trained by iteratively updating the parameters $\theta_A$ and $\theta_Z$. The authors also provide a training algorithm that incorporates observed confounders $X$, which we implemented for our experiments. During training, we used two different datasets for each of the two stages as described in in the paper.

\textbf{DeepIV} \citep{Hartford.2017}: DeepIV \citep{Hartford.2017} also assumes additive unobserved confounding as in Eq.~\eqref{eq:general_iv_ass}, but leverages the identification result \citep{Newey.2003}
\begin{equation}\label{eq:deepiv_inverse}
    \E[Y \mid X=x, Z=z] = \int h(a,x) \; \mathrm{d}F(a \mid  x, z),
\end{equation}
where $h(a,x) = f(a, x) + \E[U \mid X=x]$ is the target counterfactual prediction function. DeepIV estimates $F(a \mid x, z)$, i.e., the conditional distribution function of the treatment $A$ given observed covariates $X$ and instruments $Z$, by using neural networks. Because we consider only binary treatments, we simply implement a (tunable) feed-forward neural network with sigmoid activation function. Then, DeepIV proceeds by learning a second stage neural network to solve the inverse problem defined by Eq.~\eqref{eq:deepiv_inverse}.

\textbf{DeepGMM} \citep{Bennett.2019}: DeepGMM \citep{Bennett.2019} adopts neural networks for IV estimation inspired by the (optimally weighted) Generalized Method of Moments.
The DeepGMM estimator is defined as the solution of the following minimax game:
\begin{equation}
    \hat{\theta} \in \underset{\theta \in \Theta}{\arg \min} \; \underset{\tau \in \mathrm{T}}{\sup} \; \frac{1}{n} \sum_{i=1}^n f(z_i, \tau) (y_i - g(a_i, \theta)) - \frac{1}{4n} \sum_{i=1}^n f^2(z_i, \tau) (y_i - g(a_i, \widetilde{\theta}))^2,
\end{equation}
where $f(z_i, \cdot)$ and $g(a_i, \cdot)$ are parameterized by neural networks. As recommended in \citep{Bennett.2019}, we solve this optimization via adversarial training with the Optimistic Adam optimizer \citep{Daskalakis.2018}, where we set the parameter $\widetilde{\theta}$ to the previous value of $\theta$.

\textbf{DMLIV} \citep{Syrgkanis.2019}: DMLIV \citep{Syrgkanis.2019} assumes that the data is generated via
\begin{equation}
    Y = \tau(X) A + f(X) + U,
\end{equation}
where $\tau$ is the CATE $f$ some function of the observed covariates. First, DMLIV estimates the functions $q(X) = \E[Y \mid X]$, $h(Z, X) = \E[A \mid Z, X]$, and $p(X) = \E[A \mid X]$. Then, the CATE is learned by minimizing the loss 
\begin{equation}
    \mathcal{L}(\theta) = \sum_{i=1}  (y_i - \hat{q}(x_i) - \hat{\tau}(x_i, \theta) (\hat{h}(z_i, x_i) - \hat{p}(x_i))^2,
\end{equation}
where $\hat{\tau}(X, \cdot)$ is some model for $\tau(X)$. In our experiments, we use (tunable) feed-forward neural networks for all estimators.

\textbf{DRIV} \citep{Syrgkanis.2019}: DRIV \citep{Syrgkanis.2019} is a meta learner, originally proposed in combination with DMLIV. It requires initial estimators for $q(X)$, $p(X)$, $\pi(X) = \E[Z \mid X = x]$, and $f(X) = \E[A \cdot Z \mid X = x]$ as well as an initial CATE estimatior $\hat{\tau}_{\mathrm{init}}(X)$ (e.g., from DMLIV). The CATE is then estimated by a pseudo regression on the following doubly robust pseudo outcome:
\begin{equation}
    \hat{Y}_{\mathrm{DR}} = \hat{\tau}_{\mathrm{init}}(X) + \frac{\left(Y - \hat{q}(X) - \hat{\tau}_{\mathrm{init}}(X) (A - \hat{p}(X)) Z - \hat{\pi}(X)\right)}{\hat{f}(X) - \hat{p}(X)\hat{r}(X)}.
\end{equation}
We implement all regressions using (tunable) feed-forward neural networks. 

\underline{Comparison between DRIV vs. \frameworkname:} There are two key differences between our paper and \citep{Syrgkanis.2019}: (i) In contrast to DRIV, we showed that our MRIV is multiply robust. (ii) We derive a multiple robust convergence rate, while the rate in \citep{Syrgkanis.2019} is not robust with respect to the nuisance rates.

Ad (i): Both MRIV and DRIV perform a pseudo-outcome regression on the efficient influence function (EIF) of the ATE. The key difference: DRIV uses the doubly robust parametrization of the EIF from \citep{Okui.2012}, whereas we use the multiply robust parametrization of the EIF from \citep{Wang.2018}\footnote{For a detailed discussion on multiple robustness and the importance of the EIF parametrization, we refer to \citep{Wang.2019}, Section 4.5.}. Hence, our MRIV frameworks extends DRIV in a non-trivial way to achieve multiple robustness. Thus, our estimator is consistent in the union of \emph{three} different model specifications.\footnote{On a related note, a similar, important contribution of developing multiply robust method was recently made for the average treatment effect. Here, the estimator of \citep{Okui.2012} was extended by the estimator of \citep{Wang.2018} to allow for multi robustness. Yet, this different from our work in that it focuses on the average treatment effect, while we study the conditional average treatment effect in our paper.}

Ad (ii): Here, we compare the convergence rates from DRIV and our \frameworkname and, thereby, show the strengths of our \frameworkname. To this end, let us assume that the pseudo regression function is $\gamma$-smooth and that we use the same second-stage estimator $\hat{\mathrm{E}}_n$ with minimax rate $n^{-\frac{2\gamma}{2\gamma + p}}$ for both DRIV and \frameworkname. If the nuisance parameters $q(X)$, $p(X)$, $f(X)$, and $\pi(X)$ are $\alpha$-smooth and further are estimated with minimax rate $n^{\frac{-2\alpha}{2\alpha+p}}$, Corollary 4 from \citep{Syrgkanis.2019} states that DRIV converges with rate 
\begin{equation*}
   \E\left[\left(\hat{\tau}_{\mathrm{DRIV}}(x) - \tau(x)\right)^2\right]  \lesssim  n^{\frac{-2\gamma}{2\gamma+p}} + n^{\frac{-4\alpha}{2\alpha+p}}.
\end{equation*}
In contrast, \frameworkname assumes estimation of the nuisance parameters $\mu_0^Y(x)$ with rate $n^{\frac{-2\alpha}{2\alpha+p}}$, $\mu_0^A(x)$ and $\delta_A(x)$ with rate $n^{\frac{-2\beta}{2\beta+p}}$, and $\pi(x)$ with rate $n^{\frac{-2\delta}{2\delta+p}}$. If the initial estimator $\hat{\tau}_{\mathrm{init}}(x)$ converges with rate $r_{\tau}(n)$, our Theorem \ref{thrm:upperbound} yields the rate
\begin{equation*}
\begin{split}
    \E\left[\left(\hat{\tau}_{\mathrm{MRIV}}(x) - \tau(x)\right)^2\right] & \lesssim   n^{\frac{-2\gamma}{2\gamma+p}} + r_{\tau}(n) \left(n^{\frac{-2\beta}{2\beta+p}} + n^{\frac{-2\delta}{2\delta+p}}  \right) + 
    n^{-2\left(\frac{\alpha}{2\alpha+p}+\frac{\delta}{2\delta+p}\right)} +
    n^{-2\left(\frac{\beta}{2\beta+p}+\frac{\delta}{2\delta+p}\right)} .
\end{split}
\end{equation*}
If all nuisance parameters converge with the same minimax rate of $n^{\frac{-2\alpha}{2\alpha+p}}$, the rates of DRIV and our \frameworkname coincide. However, different to DRIV, our rate is additionally \underline{multiple robust} in spirit of Theorem~\ref{thrm:robustness}. This presents a crucial strength of our \frameworkname over DRIV: For example, if $\delta$ is small (slow convergence of $\hat{\pi}(x)$), our \frameworkname still with fast rate as long as $\alpha$ and $\beta$ are large (i.e., if the other nuisance parameters are sufficiently smooth).

\subsection{Wald estimator}
Finally, we consider the Wald estimator \citep{Wald.1940} for the binary IV setting. More precisely, we estimate the CATE components $\mu_i^Y(x)$ and $\mu_i^A(x)$ seperately and plug them into
\begin{equation}
    \tau(x) = \frac{\hat{\mu}_1^Y(x) - \hat{\mu}_0^Y(x)}{\hat{\mu}_1^A(x) - \hat{\mu}_0^A(x)}.
\end{equation}
We consider two versions of the Wald estimator:

\textbf{Linear:} We use linear regressions to estimate the $\mu_i^Y(x)$ and logistic regressions to estimate the $\mu_i^A(x)$.

\textbf{BART:} We use Bayesian additive regression trees \citep{Chipman.2010} trees to estimate the $\mu_i^Y(x)$ and random forest classifier to estimate the $\mu_i^A(x)$.

\clearpage

\section{Visualization of predicted CATEs }

We plot the predicted CATEs for the different baselines and \modelname in Fig.~\ref{fig:model_fits} (for $n = 3000$). As expected, the linear methods (2SLS and linear Wald) are not flexible enough to provide accurate CATE estimates. We also observe that the curve of \modelname without \frameworkname is quite wiggly, i.e., the estimator has a relatively large variance. This variance is reduced when the full \modelname is applied. As a result, curve is much smoother. This is reasonable because \frameworkname does not estimate the CATE components individually, but estimates the CATE directly via the Stage~2 pseudo outcome regression. Overall, this confirms the superiority of our proposed framework. 

\begin{figure}[h]
\centering
\includegraphics[width=1\linewidth]{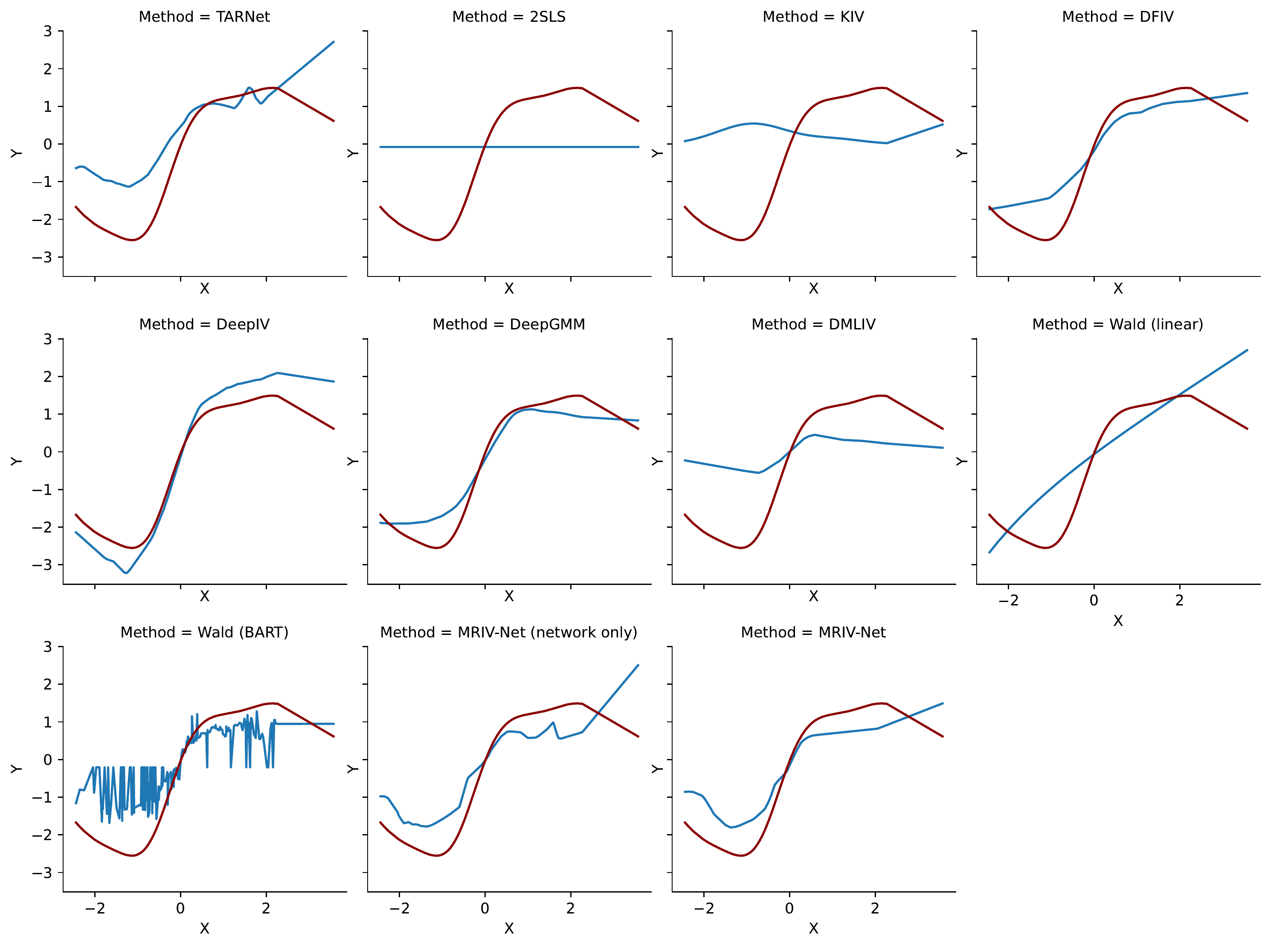}
\caption{Predicted CATEs (blue) and oracle CATE (red) for different baselines.}
\label{fig:model_fits}
\end{figure}

\clearpage

\section{Implementation details and hyperparameter tuning}\label{app:hyper}

\textbf{Implementation details for deep learning models:}
To make the performance of the deep learning models comparable, we implemented all feed-forward neural networks (including \modelname) as follows: We use two hidden layers with RELU activation functions. We also incorporated a dropout layer for each hidden layer. We trained all models with the Adam optimizer \citep{Kingma.2015} using 100 epochs. Exceptions are only DFIV and DeepGMM, where we used 200 epochs for training, accounting for slower convergence of the respective (adversarial) training algorithms. For DeepGMM, we further used Optimistic Adam \citep{Daskalakis.2018} as in the original paper.

\textbf{Training times:} We report the approximate times needed to train the deep learning models on our simulated data with $n=5000$ in Table~\ref{t:times}. For training, we used an AMD Ryzen Pro 7 CPU. Compared to DMLIV and DRIV, the training of \modelname is faster because only a single neural network is trained.

\begin{table}[h]
\caption{Training times for deep learning models (in seconds).}
\centering
\label{t:times}
\resizebox{\columnwidth}{!}{%
\begin{tabular}{cccccccc}
\noalign{\smallskip} \toprule \noalign{\smallskip}
TARNet & TARNet + DR & DFIV & DeepIV & DeepGMM & DMLIV & DMLIV + DRIV & \modelname \\
\midrule
$\sim$10.62 & $\sim$28.57 &$\sim$164.98&$\sim$30.21& $\sim$17.31 & $\sim$74.98& $\sim$91.12&$\sim$32.20 \\
\bottomrule
\end{tabular}}
\end{table}

\textbf{Hyperparameter tuning:} We performed hyperparameter tuning for all deep learning models (including \modelname), KIV, and the BART Wald estimator on all datasets. For all methods except KIV and DFIV, we split the data into a training set (80\%) and a validation set (20\%). We then performed 40 random grid search iterations and chose the set of parameters that minimized the respective training loss on the validation set. In particular, the tuning procedure was the same for all baselines, which ensures that the performance gain of \modelname is due to the method itself and not due to larger flexibility. Exceptions are only KIV and DFIV, for which we implemented the customized hyperparameter tuning algorithms proposed in \citep{Singh.2019} and \citep{Xu.2021} to ensure consistency with prior literature. For the meta learners (DR-learner, DRIV, and \frameworkname), we first performed hyperparameter tuning for the base methods and nuisance models, before tuning the pseudo outcome regression neural network by using the input from the tuned models. The tuning ranges for the hyperparameter are shown in Table~\ref{t:hyper}. These include both the hyperparameter rangers shared across all neural networks and the model-specific hyperparameters. For reproducibility purposes, we publish the selected
hyperparameters in our GitHub project as \emph{.yaml} files.

\begin{table}[h]
\caption{Hyperparameter tuning ranges.}
\centering
\label{t:hyper}
\resizebox{\columnwidth}{!}{%
\begin{tabular}{lll}
\noalign{\smallskip} \toprule \noalign{\smallskip}
\textsc{Model} & \textsc{Hyperparameter} & \textsc{Tuning range} \\
\midrule
\midrule
Feed-forward neural networks&Hidden layer size(es) & $p$, $5p$, $10p$, $20p$, $30p$ (simulated data) \\
(Shared parameter ranges&& $p$, $3p$, $5p$, $8p$, $10p$ (OHIE) \\
for all deep learning baselines)&Learning rate & $0.0001$, $0.0005$, $0.001$, $0.005$, $0.01$\\
&Batch size &$64$, $128$, $256$ \\
&Dropout probability & $0$, $0.1$, $0.2$, $0.3$\\
\midrule
KIV & $\lambda$ (Ridge penalty first stage) & 5, 6, 7, 8, 9, 10, 12\\
 & $\xi$ (Ridge penalty second stage) & 5, 6, 7, 8, 9, 10, 12\\
DFIV& $\lambda_1$ (Ridge penalty first stage) & 0.0001, 0.001, 0.01, 0.1 (simulated data) \\
 &  & 0.01, 0.05, 0.1 (OHIE) \\
& $\lambda_2$ (Ridge penalty second stage) & 0.0001, 0.001, 0.01, 0.1 (simulated data) \\
 &  & 0.01, 0.05, 0.1 (OHIE) \\
DeepGMM & $\lambda_f$ (learning rate multiplier) &0.5, 1, 1.5, 2, 5\\
Wald (BART) & Number of trees (BART) & 20, 30, 40, 50\\
 & Number of trees (Random forest classifier) & 20, 30, 40, 50\\
\bottomrule
\multicolumn{3}{l}{$p =$ network input size}
\end{tabular}}
\end{table}

\textbf{Hyperparameter robustness checks:}
We also investigate the robustness of \modelname with respect to hyperparameter choice. To to this, we fix the optimal hyperparameter constellation for our simulated data for $n=3000$ and perturb the hidden layer sizes, learning rate, dropout probability, and batch size. The results are shown in Fig.~\ref{fig:hyper_robust}. We observe that the RMSE only changes marginally when perturbing the different hyperparameters, indicating that our method is to a certain degree robust against hyperparameter misspecification. Furthermore, our results indicate that the performance improvement of \modelname over the baselines observed in our experiments is not due to hyperparameter tuning, but to our method itself.

\begin{figure}[h]
\centering
\includegraphics[width=0.6\linewidth]{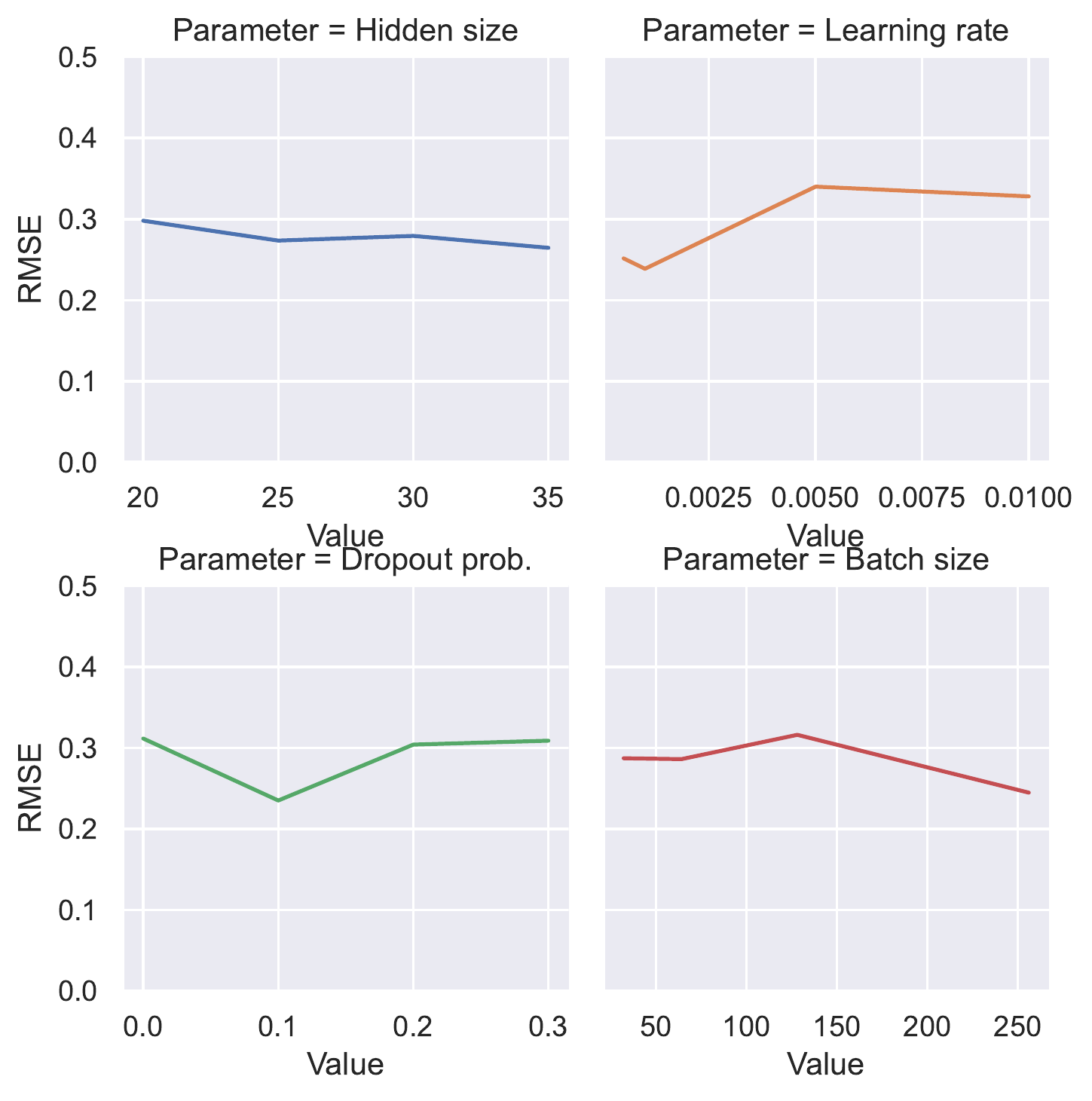}
\caption{Robustness checks for different hyperparameters of \modelname.}
\label{fig:hyper_robust}
\end{figure}

\clearpage

\section{Results for semi-synthetic data}\label{app:results}

In the main paper, we evaluated \modelname both on synthetic and real-world data. Here, we provide additional results by constructing a semi-synthetic dataset on the basis of OHIE. It is common practice in causal inference literature to use semi-synthetic data for evaluation, because it combines advantages of both synthetic and real-world data. On the one hand, the real-world data part ensures that the data distribution is realistic and matches those in practice. On the other hand, the counterfactual ground-truth is still available, which makes it possible to measure the performance of CATE methods.

We construct our semi-synthetic data as follows: First, we extract the covariates $X \in \R^5$ and instruments $Z \in \{0,1\}$ of our OHIE dataset from Sec.~\ref{app:real}. Then, we construct the treatment components $\mu_i^A(x)$ via
\begin{equation}
    \mu_1^A(X) = 0.3 \cdot \sigma(X_1) + 0.7 \quad \text{ and } \mu_0^A(X) = 0.3 \cdot \sigma(X_1) ,
\end{equation}
where $X_1$ is the (standardized) age and $\sigma(\cdot)$ is the sigmoid function. The outcome components are constructed via
\begin{equation}
    \mu_1^Y(X) = 0.5 X_1^2 + \sum_{i=2}^5 X_i^2 \quad \text{ and } \mu_0^Y(X) = - 0.5 X_1^2 + \sum_{i=2}^5 X_i^2.
\end{equation}
We then sample treatments $A$ and outcomes $Y$ as in Eq.~\eqref{eq:sim_treat} and Eq.~\eqref{eq:sim_outcome}. Lemma~\ref{lem:sim} ensures that $\mu_i^Y(X) = \E[Y \mid Z = i, X]$ and $\mu_i^A(X) = \E[A \mid Z = i, X]$.

Given the above, the oracle CATE becomes
\begin{equation}
    \tau(X) = \frac{X_1^2}{0.7} .
\end{equation}
Note that $\tau(X)$ is sparse in the sense that it only depends on age, while the outcome components depend on all five covariates. Following our theoretical analysis in Sec.~\ref{app:sparsity}, \modelname should thus outperform methods that aim at estimating the components directly. This is confirmed in Table~\ref{tab:semi_synth}, where we show the results for all baselines and \modelname on the semi-synthetic data. Indeed, we observe that \modelname outperforms all other baselines, confirming both the superiority of our method as well as our theoretical results under sparsity assumptions from Sec.~\ref{app:sparsity}.

\begin{table}[h]
\caption{Results for semi-synthetic data.}
\label{tab:semi_synth}
\centering
\resizebox{\columnwidth}{!}{%
\begin{tabular}{lccc}
\toprule
{Method} & {$n = 3000$} & {$n = 5000$} & {$n = 8000$} \\
\midrule
\textsc{(1) Standard ITE} & &  & \\
\quad TARNet \citep{Shalit.2017} &$1.66 \pm 0.11$& $1.58 \pm 0.07$ & $1.57 \pm 0.11$\\
\quad TARNet + DR \citep{Shalit.2017, Kennedy.2022d} &$1.31 \pm 0.28$& $1.22 \pm 0.37$ & $1.12 \pm 0.15$\\
\midrule
\textsc{(2) General IV} & &  & \\
\quad 2SLS \citep{Wooldridge.2013}&$1.34 \pm 0.06$ &$1.31 \pm 0.03$  & $1.32 \pm 0.02$ \\
\quad KIV \citep{Singh.2019}&$1.97 \pm 0.10$ & $1.92 \pm 0.05$ &  $1.93 \pm 0.05$\\
\quad DFIV \citep{Xu.2021}&$1.67 \pm 0.44$ & $1.63 \pm 0.47$ & $1.45 \pm 0.17$ \\
\quad DeepIV \citep{Hartford.2017}&$1.24 \pm 0.26$ & $0.99 \pm 0.22$ & $0.84 \pm 0.19$\\
\quad DeepGMM \citep{Bennett.2019}&$1.39 \pm 0.03$ &$1.37 \pm 0.16$  & $1.18 \pm 0.16$ \\
\quad DMLIV \citep{Syrgkanis.2019}&$2.12 \pm 0.10$ & $2.09 \pm 0.09$ &  $2.02 \pm 0.11$\\
\quad DMLIV + DRIV \citep{Syrgkanis.2019}&$1.22 \pm 0.10$ & $1.18 \pm 0.19$ &  $1.00 \pm 0.08$\\
\midrule
\textsc{(3) Wald estimator \citep{Wald.1940}} & &  & \\
\quad Linear &$1.42 \pm 0.24$ & $1.28 \pm 0.07$ &  $1.32 \pm 0.07$\\
\quad BART &$1.48 \pm 0.24$ &$1.29 \pm 0.04$ &$1.06 \pm 0.13$\\ \bottomrule \noalign{\smallskip}
\modelname (network only) &$1.11 \pm 0.15$ & $0.84 \pm 0.14$ & $0.95 \pm 0.21$\\
\modelname (ours) &$\boldsymbol{0.71 \pm 0.24}$ & $\boldsymbol{0.75 \pm 0.18}$ & $\boldsymbol{0.78 \pm 0.26}$\\

\bottomrule
\multicolumn{4}{l}{Reported: RMSE (mean $\pm$ standard deviation). Lower $=$ better (best in bold)}
\end{tabular}}

\vspace{-0.5cm}
\end{table}

\clearpage

\section{Results for cross-fitting}\label{app:results_cf}

Here, we repeat our experiments from the main paper but now make use of \emph{cross-fitting}. Recall that, in Theorem~\ref{thrm:upperbound}, we assume that the nuisance parameter estimation and the pseudo-outcome regression are performed on three independent samples. We now address this through \emph{cross-fitting}. To this end, our aim is to show that our proposed \frameworkname framework is again superior. 

For \frameworkname, we proceeded as follows: We split the sample $\mathcal{D}$ into three equally sized samples $\mathcal{D}_1$, $\mathcal{D}_2$, and $\mathcal{D}_3$. We then trained $\hat{\tau}_{init}(x)$, $\hat{\mu}_0^Y(x)$, and $\hat{\mu}_0^A(x)$ on $\mathcal{D}_1$, $\hat{\delta}_A(x)$ and $\hat{\pi}(x)$ on $\mathcal{D}_2$, and performed the pseudo-outcome regression on $\mathcal{D}_3$. Then, we repeated the same training procedure two times, but performed the pseudo-outcome regression on $\mathcal{D}_2$ and $\mathcal{D}_1$. Finally, we averaged the resulting three CATE estimators. For DRIV, we implemented the cross-fitting procedure described in \citep{Syrgkanis.2019}. For the DR-learner, we followed \citep{Kennedy.2022d}. 

The results are in Table~\ref{tab:results_sim_cf}. Importantly, the results confirm the effectiveness of our proposed \frameworkname. Overall, we find that our proposed \frameworkname outperforms DRIV for the vast majority of base methods when performing cross-fitting. Furthermore, \modelname is highly competitive even when comparing it with the cross-fitted estimators. This shows that our heuristic to learn separate representations instead of performing sample splits works in practice. In sum, the results confirm empirically that our \frameworkname is superior. 

\vspace{0.5cm}

\begin{table}[h]
\caption{Results for base methods with different meta-learners (i.e., DRIV, and our \frameworkname) using cross-fitting and results for MRIV-Net without cross-fitting.}
\label{tab:results_sim_cf}
\centering
\resizebox{\columnwidth}{!}{%
\begin{tabular}{lcccccc}
\noalign{\smallskip} \toprule \noalign{\smallskip}
& \multicolumn{2}{c}{$n = 3000$} & \multicolumn{2}{c}{$n = 5000$} & \multicolumn{2}{c}{$n = 8000$} \\
\cmidrule(lr){2-3} \cmidrule(lr){4-5} \cmidrule(lr){6-7}
\backslashbox{Base methods}{Meta-learners} & DRIV & \frameworkname (ours) & DRIV & \frameworkname (ours) & DRIV & \frameworkname (ours)\\
\midrule 
\textsc{(1) Standard ITE} & & & \\
\quad TARNet \citep{Shalit.2017} &$\boldsymbol{0.30 \pm 0.02}$ & $0.36 \pm 0.16$ & $0.18 \pm 0.06$ &$\boldsymbol{0.16 \pm 0.03}$ & $0.21 \pm 0.08$ & $\boldsymbol{0.13 \pm 0.04}$\\
\quad TARNet + DR-learner \citep{Shalit.2017, Kennedy.2022d} &\multicolumn{2}{c}{$0.85 \pm 0.11$} & \multicolumn{2}{c}{$0.66 \pm 0.08$} & \multicolumn{2}{c}{$0.67 \pm 0.12$}\\
\midrule
\textsc{(2) General IV} & & & \\
\quad 2SLS \citep{Wooldridge.2013}&$0.42 \pm 0.11$ & $\boldsymbol{0.33 \pm 0.09}$ & $\boldsymbol{0.20 \pm 0.07}$ & $0.23 \pm 0.11$ & $0.24 \pm 0.10$& $\boldsymbol{0.14 \pm 0.02}$\\
\quad KIV \citep{Singh.2019}& $0.47 \pm 0.18$& $\boldsymbol{0.45 \pm 0.15}$ & $0.20 \pm 0.06$ & $\boldsymbol{0.19 \pm 0.08}$& $0.22 \pm 0.04$& $\boldsymbol{0.15 \pm 0.03}$\\
\quad DFIV \citep{Xu.2021}&$0.35 \pm 0.05$ & $\boldsymbol{0.28 \pm 0.09}$ & $0.22 \pm 0.10$ &$\boldsymbol{0.18 \pm 0.08}$ & $0.24 \pm 0.12$ & $\boldsymbol{0.16 \pm 0.04}$ \\
\quad DeepIV \citep{Hartford.2017}& $\boldsymbol{0.38 \pm 0.09}$& $0.44 \pm 0.16$ & $0.20 \pm 0.07$ & $\boldsymbol{0.19 \pm 0.07}$ & $0.20 \pm 0.08$ & $\boldsymbol{0.12 \pm 0.02}$\\
\quad DeepGMM \citep{Bennett.2019}&$\boldsymbol{0.42 \pm 0.09}$ & $\boldsymbol{0.42 \pm 0.16}$ & $\boldsymbol{0.19 \pm 0.04}$ &$\boldsymbol{0.19 \pm 0.07}$ & $0.22 \pm 0.06$ & $\boldsymbol{0.13 \pm 0.02}$ \\
\quad DMLIV \citep{Syrgkanis.2019} &$\boldsymbol{0.44 \pm 0.09}$ & $0.46 \pm 0.16$ & $0.21 \pm 0.04$ &$\boldsymbol{0.19 \pm 0.07}$ & $0.21 \pm 0.05$& $\boldsymbol{0.14 \pm 0.02}$\\
\midrule
\textsc{(3) Wald estimator \citep{Wald.1940}}& & & & & & \\
\quad Linear &$0.47 \pm 0.23$ & $\boldsymbol{0.36 \pm 0.12}$ & $0.24 \pm 0.05$ & $\boldsymbol{0.20 \pm 0.08}$ & $0.22 \pm 0.05$ & $\boldsymbol{0.15 \pm 0.02}$\\
\quad BART &$0.43 \pm 0.12$ & $\boldsymbol{0.39 \pm 0.12}$ & $0.14 \pm 0.05$ &$\boldsymbol{0.13 \pm 0.05}$ & $0.23 \pm 0.08$ & $\boldsymbol{0.15 \pm 0.02}$\\
\midrule
\modelname {\textbackslash}w network only (ours)&$0.35 \pm 0.12$ & $\boldsymbol{0.26 \pm 0.11}$& $0.19 \pm 0.13$ & $\boldsymbol{0.15 \pm 0.03}$& $0.18 \pm 0.08$& $\boldsymbol{0.13 \pm 0.03}$\\

\bottomrule
\multicolumn{7}{l}{Reported: RMSE (mean $\pm$ standard deviation). Lower $=$ better (best in bold)}
\end{tabular}}%

\end{table}

\clearpage

\section{Further experimental results on real-world data}\label{app:real_furtherexp}
In Section~\ref{sec:exp_real}, we estimated the ITE on the OHIE data and visualized the treatment heterogeneity with respect to \emph{age} and \emph{gender}. In this section, we provide additional results and also visualize the heterogeneity with respect to age as well as \emph{additional covariates}. These are: (1)~the number of emergency visits a patient has in its history before signing up for the lottery and (2)~the language spoken by the patient (English or other). We fixed the gender to ``female''.

For (1), we plot the estimated ITE for three different age groups over the number of emergency visits. The results are shown in Fig.~\ref{fig:real_visits}.
\begin{figure}[h]
\centering
\includegraphics[width=1\linewidth]{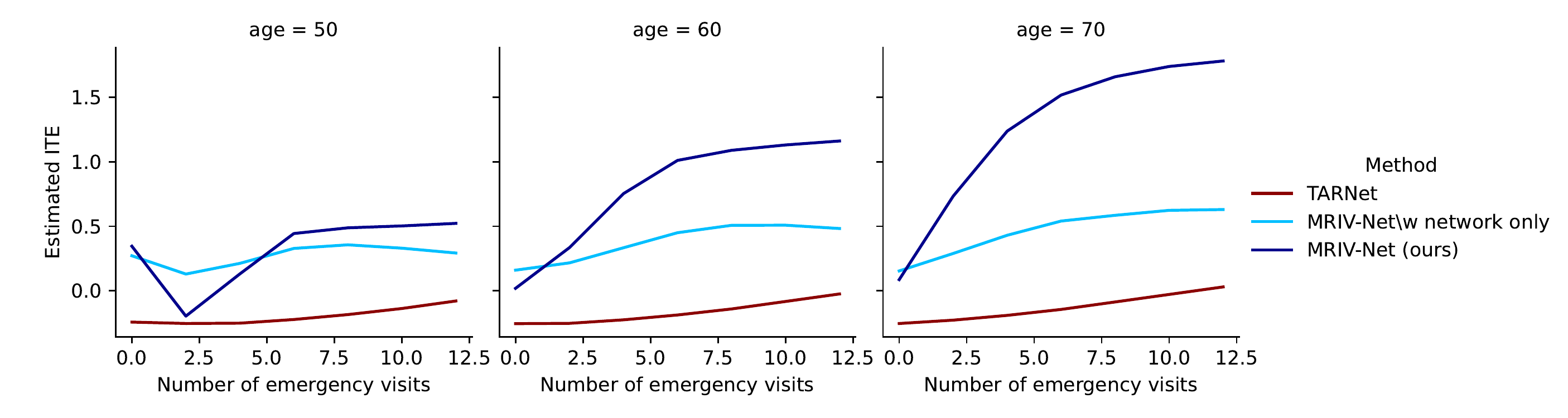}
\caption{Estimated treatment effects for different age and number of emergency visits.}
\label{fig:real_visits}
\end{figure}

We observe that all methods tend to estimate a larger effect for individuals who had more emergency visits in their patient history. However, the IV methods (in particular our \modelname) estimate a much larger effect for patients with many visits. In contrast to the other methods, \modelname also estimates larger effects for older than for younger patients. The results provided by \modelname seem intuitive, as older patients with a history of emergency visits should be exposed to higher health-related risks, thus benefiting from health insurance. The fact that TARNet consistently estimates small (and even negative) effects could be an indicator of bias due to unobserved confounding.

For (2), we plot the estimated ITE for three different age groups over the spoken language. The results are shown in Fig.~\ref{fig:real_language}.

\begin{figure}[h]
\centering
\includegraphics[width=1\linewidth]{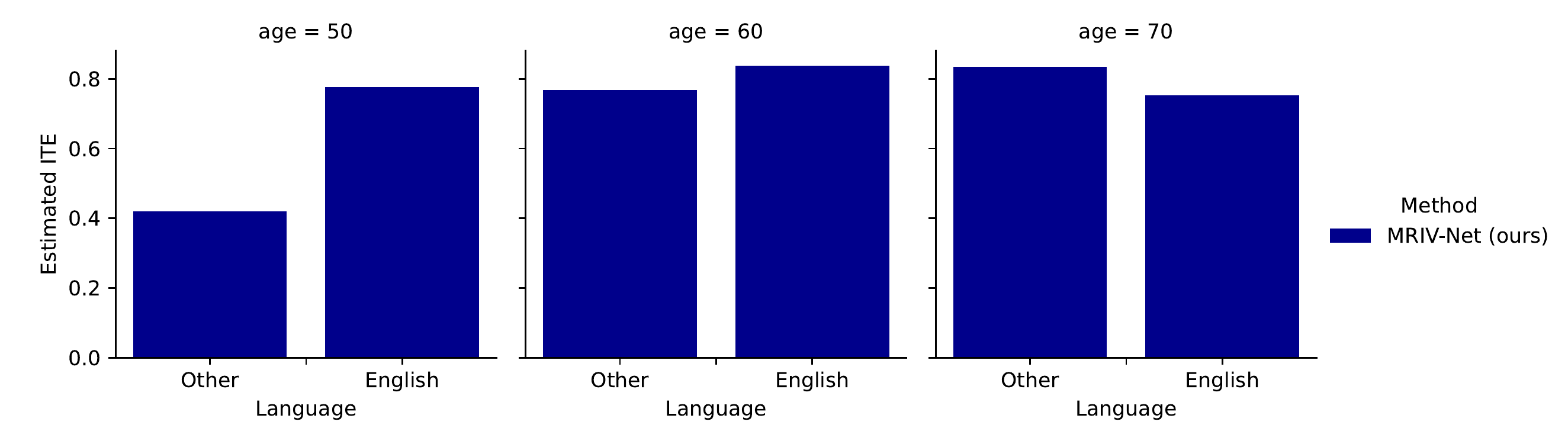}
\caption{Estimated treatment effects for different age and language}
\label{fig:real_language}
\end{figure}
For patients of age 50, our \modelname estimates a higher effect for the English-speaking patients. Interestingly, for older patients, the estimated effect increases also for non-English speaking patients.

\end{document}